\documentclass[12pt,reqno]{amsart}
\pdfoutput=1 
\usepackage{amsmath,bbm}
\usepackage{latexsym}
\usepackage{amsfonts}
\usepackage{amssymb}
\usepackage{color}
\usepackage{graphicx}
\usepackage{url}
\usepackage{enumerate}
\usepackage{tikz}
\usetikzlibrary{shadings, intersections, calc, plotmarks}

\usepackage{geometry}
\xdefinecolor{tumblue}     {RGB}{0,101,189}
\xdefinecolor{tumgreen}    {RGB}{162,173,  0}
\xdefinecolor{tumred}      {RGB}{229, 52, 24}
\xdefinecolor{tumivory}    {RGB}{218,215,203}
\xdefinecolor{tumorange}   {RGB}{227,114, 34}
\xdefinecolor{tumlightblue}{RGB}{152,198,234}

\newtheorem{proposition}{Proposition}
\newtheorem{theorem}{Theorem}
\newtheorem*{theorem*}{Theorem}
\newtheorem{lemma}{Lemma}
\newtheorem{corollary}{Corollary}
\newtheorem*{corollary*}{Corollary}
\newtheorem{assumption}{Assumption}

\newtheorem{definition}{Definition}

\newcommand{\Ga}{\begin{tikzpicture}[x=0.44cm,y=0.3cm]%
\draw[thick] (0,0) -- (1,0);%
\draw[thick] (0,0.5) -- (1,0.5);%
\draw[thick] (0,1) -- (1,1);%
\end{tikzpicture}}
\newcommand{\Gb}{\begin{tikzpicture}[x=0.44cm,y=0.3cm]%
\draw[thick] (0,0) -- (1,0.5);%
\draw[thick] (0,0.5) -- (1,0);%
\draw[thick] (0,1) -- (1,1);%
\end{tikzpicture}}
\newcommand{\Gc}{\begin{tikzpicture}[x=0.44cm,y=0.3cm]%
\draw[thick] (0,0) -- (1,0);%
\draw[thick] (0,0.5) -- (0.4,0.5) -- (0.4,1) -- (0,1);%
\draw[thick] (1,0.5) -- (0.6,0.5) -- (0.6,1) -- (1,1);%
\end{tikzpicture}}
\newcommand{\Gd}{\begin{tikzpicture}[x=0.44cm,y=0.3cm]%
\draw[thick] (0,0) -- (1,0);%
\draw[thick] (0,0.5) --  (1,1);%
\draw[thick] (0,1) -- (1,0.5);%
\end{tikzpicture}}
\newcommand{\Ge}{\begin{tikzpicture}[x=0.44cm,y=0.3cm]%
\draw[thick] (0,0) -- (1,0);%
\draw[thick] (0.4,0) --  (0.4,1) -- (0,1);%
\draw[thick] (0,0.5) -- (0.4,0.5);%
\draw[thick] (1,0.5) --  (0.6,0.5) -- (0.6,1) -- (1,1);%
\end{tikzpicture}}
\newcommand{\Gf}{\begin{tikzpicture}[x=0.44cm,y=0.3cm]%
\draw[thick] (0,0) -- (1,0);%
\draw[thick] (0.6,0) --  (0.6,1) -- (1,1);%
\draw[thick] (0.6,0.5) -- (1,0.5);%
\draw[thick] (0,0.5) --  (0.4,0.5) -- (0.4,1) -- (0,1);%
\end{tikzpicture}}
\newcommand{\Gg}{\begin{tikzpicture}[x=0.44cm,y=0.3cm]%
\draw[thick] (0,0.5) -- (1,0.5);%
\draw[thick] (1,0) -- (0,1);%
\draw[thick] (0,0) -- (1,1);%
\end{tikzpicture}}

\newcommand{\id}{{\rm{id}}} %identity

\newcommand{\R}{\mathbbm{R}}
\newcommand{\C}{\mathbbm{C}}

\newcommand{\F}{\mathbbm{F}}
\newcommand{\Q}{\mathbbm{Q}}
\newcommand{\cP}{\mathcal{P}}
\newcommand{\cT}{\mathcal{T}}
\newcommand{\cC}{\mathcal{C}}
\newcommand{\cS}{\mathcal{S}}
\newcommand{\cA}{\mathcal{A}}
\newcommand{\cE}{\mathcal{E}}

\newcommand{\cM}{\mathcal{M}}
\newcommand{\1}{\mathbbm{1}}
\newcommand{\DD}{\hat{\Delta}}
\def\>{{\rangle}}
\def\<{{\langle}}
\newcommand{\be}{\begin{equation}}
\newcommand{\ee}{\end{equation}}
\newcommand{\bea}{\begin{eqnarray}}
\newcommand{\eea}{\end{eqnarray}}

\newcommand{\ket}[1]{|#1\rangle} %ket
\newcommand{\bra}[1]{\langle#1|} %bra
\newcommand{\kb}[2]{|#1\rangle\langle#2|} %ketbra
\newcommand{\bk}[2]{\langle #1 | #2 \rangle} %braket
\newcommand{\tr}[1]{\mathrm{tr}\left[#1\right]} %trace
\newcommand{\Tr}{\mathrm{tr}}
\newcommand{\norm}[1]{\left\lVert #1 \right\rVert}

\begin{document}

\title[Information-Disturbance Tradeoff]{Universality and  Optimality in the \vspace*{5pt} \\ Information-Disturbance Tradeoff}

\author[Hashagen]{Anna-Lena K. Hashagen$^1$}
\author[Wolf]{Michael M. Wolf$^{1,2}$}
\address{$^1$ Department of Mathematics, Technical University of Munich}
\address{$^2$ Kavli Institute for Theoretical Physics, University of California, Santa Barbara (Aug - Dec, 2017)}

\begin{abstract}We investigate the tradeoff between the quality of an approximate version of a given measurement and the disturbance it induces in the measured quantum system. We prove that if the target measurement is a non-degenerate von Neumann measurement, then the optimal tradeoff can always be achieved within a two-parameter family of quantum devices that is independent of the chosen distance measures. This form of almost universal optimality holds under mild assumptions on the distance measures such as convexity and basis-independence, which are satisfied for all the usual cases that are based on norms, transport cost functions, relative entropies, fidelities, etc. for both worst-case and average-case analysis. We analyze the case of the cb-norm (or diamond norm) more generally for which we show dimension-independence of the derived optimal tradeoff for general von Neumann measurements. A SDP solution is provided for general POVMs and shown to exist for arbitrary convex semialgebraic distance measures.  
\end{abstract}

\maketitle
\tableofcontents
\newpage

%%%%%%%%%%%%%%%%%%%%%%%%%%%%%%%%%%%%%%%%
\section{Introduction}\label{sec:intro}
%%%%%%%%%%%%%%%%%%%%%%%%%%%%%%%%%%%%%%%%
The idea that measurements inevitably disturb a quantum system is so much folklore and so deeply routed in the foundations of quantum mechanics  that it is difficult to trace back historically. It is certainly present in Heisenberg's original exposition of the uncertainty relation. However, it only became amenable to mathematical analysis after the `projection postulate' was replaced by a more refined theory of the quantum measurement process ~\cite{Davies_Lewis_1970, Lueders_1950}. 
With the emergence of the field of quantum information theory, the interest in a quantitative analysis of the information-disturbance tradeoff has intensified. At the same time, it became an issue of practical significance for many quantum information processing tasks, most notably for quantum cryptography \cite{BB_1984, Ekert_1991, Fusch_Peres_1996, Fuchs_2005}.

In the last two decades numerous papers derived quantitative bounds on the disturbance induced by a quantum measurement. A coarse way to categorize the existing approaches is depending on whether or not there are reference measurements w.r.t. which information gain on one side and disturbance on the other side are quantified. In \cite{Martens_1992, Ozawa_2003, Ozawa_2004, HeinosaariWolf_2010, Watanabe_2011, Ipsen_2013, Busch_Lahti_Werner_2013, Busch_Lahti_Werner_2014, Branciard_2013, Buscemi_2014, Coles_2015, Schwonnek_Reeb_Werner_2016, Renes_2017} disturbance and information gain are both considered w.r.t. reference measurements. In \cite{Banaszek_2001, Barnum_2001, Maccone_2007, Kretschmann_2008, Buscemi_Hayashi_Horodecki_2008, Buscemi_Horodecki_2009, Bisio_Chiribella_DAriano_Perinotti_2010, Shitara_Kuramochi_Ueda_2016}, in contrast, no reference observable is used on either side.  In the present paper, we follow an intermediate route: we consider the performed measurement as an approximation of a given reference measurement, but we quantify the disturbance without specifying a second observable. 

Another way of classifying previous works is in terms of the measures that are used to mathematically formalize and quantify disturbance and information gain: for instance, \cite{Martens_1992, Barnum_2001, Buscemi_Horodecki_2009, Buscemi_2014, Maccone_2007, Coles_2015} use various entropic measures, \cite{Kretschmann_2008, Ipsen_2013, Renes_2017} use norm-based measures, \cite{Banaszek_2001, Barnum_2001, Buscemi_Horodecki_2009, Bisio_Chiribella_DAriano_Perinotti_2010} use fidelities, \cite{Watanabe_2011, Shitara_Kuramochi_Ueda_2016} use Fisher information, and \cite{Busch_Lahti_Werner_2013, Schwonnek_Reeb_Werner_2016} use transport-cost functions. Many other measures are conceivable and most of them come in two flavors: a worst-case and an average-case variant, where the latter again calls for the choice of an underlying distribution. 

A central point of the present work is to show that the information-disturbance problem has a core that is largely independent of the measures chosen. More specifically, we prove the existence of a small set of devices that are (almost) universally optimal independent of the chosen measures, as long as these  exhibit a set of elementary properties that are shared by the vast majority of distance measures found in the literature. Based on this universality result, we then derive optimal tradeoff bounds for specific choices of measures. These include the diamond norm and its classical counterpart the total variation distance. In this case, the reachability of the optimal tradeoff has been demonstrated experimentally in a parallel work \cite{Knips_2018}.
\newpage
\paragraph{\bf Organization of the paper.}

 Sec.~\ref{sec:sum} starts off with introducing the setup and summarizes the paper's main results. In Sec.~\ref{sec:dist}, we discuss distance measures that quantify the measurement error and the disturbance caused to the system. We give a brief overview of  common measures found in the literature that fulfill the assumptions we make, necessary to derive the universality theorem. In Sec.~\ref{sec:vN1}, for the case of a non-degenerate von Neumann target measurement, we derive a universal two-parameter family of optimal devices that yield the best information-disturbance tradeoff. In Sec.~\ref{sec:vN2}, still for the case of a non-degenerate von Neumann target measurement, we use the universal optimal devices derived in the previous section to compute the optimal tradeoff for a variety of distance measures. In the special case where we consider the diamond norm for quantifying disturbance, we derive the optimal tradeoff also for the case of degenerate von Neumann target measurements. In the last section, Sec.~\ref{sec:SDP}, we show that the optimal tradeoff can always be represented as a SDP if the distance measures under consideration are convex semialgebraic. We give the explicit SDP that represents the tradeoff between the diamond norm and the worst-case $l_\infty$-distance and apply it to the special case of qubit as well as qutrit SIC POVMs.
%%%%%%%%%%%%%%%%%%%%%%%%%%%%%%%%%%%%%%%%
\section{Summary}\label{sec:sum}
%%%%%%%%%%%%%%%%%%%%%%%%%%%%%%%%%%%%%%%%
This section will briefly introduce some notation, specify the considered setup, and summarize the main results. More details and proofs will then be given in the following sections.\vspace*{5pt}
\paragraph{\bf Notation.} Throughout we will consider finite dimensional Hilbert spaces $\C^d$, write  $\cM_d$ for the set of complex $d\times d$ matrices and $\cS_d\subseteq \cM_d$ for the subset of density operators, usually denoted by $\rho$. An $m$-outcome measurement on this space will be described by a \emph{positive operator valued measure} (POVM) $E=(E_1,\ldots,E_m)$ whose elements $E_i\in\cM_d$ are positive semidefinite and sum up to the identity operator $\sum_{i=1}^m E_i=\1$. The set of all such POVM's will be denoted by $\cE_{d,m}$ and we will set $\cE_d:=\cE_{d,d}$. We will call $E$ a \emph{von Neumann measurement} if the $E_i$'s are mutually orthogonal projections and further call it \emph{non-degenerate} if those are one-dimensional, i.e., characterized by an orthonormal basis. A completely positive, trace-preserving linear map will be called a \emph{quantum channel} and the set of quantum channels from $\cM_d$ into $\cM_d$ will be denoted by $\cT_d$. \vspace*{5pt}

\paragraph{\bf Setup.} We will fix a \emph{target measurement} $E\in\cE_{d,m}$ and investigate the tradeoff between the quality of an approximate measurement of $E$, say by $E'\in\cE_{d,m}$, and the disturbance the measurement process induces in the system. The evolution of the latter will be described by some channel $T_1\in\cT_d$.  To this end, we will have to choose two suitable functionals $E'\mapsto\delta(E')$ and $T_1\mapsto \Delta(T_1)$ that quantify the deviation of $E'$ and $T_1$ from the target measurement $E$ and the ideal channel $\id$, respectively. 

For a given triple $(E,\delta,\Delta)$ the question will then be: what is the accessible region in the $\delta-\Delta$-plane when running over all possible measurement devices and, in particular, what is the optimal tradeoff curve and how can it be achieved?

Clearly, $E'$ and $T_1$ are not independent. The framework of \emph{instruments} allows to describe all pairs ($E'$, $T_1$) that are compatible within the rules of quantum theory. An \emph{instrument} assigns to each possible outcome $i$ of a measurement a completely positive map $I_i:\cM_d\rightarrow\cM_d$ so that the corresponding POVM element is $E_i':=I_i^*(\1)$ and the evolution of the remaining quantum system is governed by $T_1:=\sum_{i=1}^m I_i$. Normalization requires that this sum is trace-preserving.\vspace*{5pt}

\paragraph{\bf Main results.} There are zillions of possible choices for the  measures $\Delta$ and $\delta$. If one had to choose one pair that stands out for operational significance this would probably be the \emph{diamond norm} and its classical counterpart, the \emph{total variational distance} (defined and discussed in Sec.~\ref{sec:dist} and Sec.~\ref{sec:vN2}). One of our results is the derivation of the optimal tradeoff curve for this pair (Thm.~\ref{thm:TVdiamond} in Sec.~\ref{sec:gvN}): 
\begin{theorem*}[Total variation - diamond norm tradeoff]
If an instrument is considered approximating a (possibly degenerate) von Neumann measurement with $m$ outcomes, then the worst-case total variational distance $\delta_{TV}$ and the diamond norm distance $\Delta_\diamond$ satisfy
\be \delta_{TV}\geq \left\{\begin{array}{ll}\frac{1}{2m}\left(\sqrt{(2-\Delta_\diamond)(m-1)}-\sqrt{\Delta_\diamond} \right)^2&\ \text{if }\ \Delta_\diamond\leq 2-\frac{2}{m},\\ 0 &\ \text{if }\ \Delta_\diamond > 2-\frac{2}{m}.
\end{array}\right.\label{eq:optdiatv}\ee
The inequality is tight in the sense that for every choice of the von Neumann measurement there is an instrument achieving equality.
\end{theorem*}
Note that the tradeoff depends solely on the number $m$ of outcomes and is independent of the dimension of the underlying Hilbert space (apart from $d\geq m$). Also note that the accessible region shrinks with increasing $m$ and in the limit $m\rightarrow\infty$ becomes a triangle, determined by $\delta_{TV}\geq 1-\Delta_{\diamond}/2$.

In Sec.~\ref{sec:vN2} we derive similar results for the worst-case as well as average-case fidelity and trace-norm. In all cases, the bounds are tight and we show how the optimal tradeoff can be achieved. Instead of going through these and more examples one-by-one we follow a different approach. We provide a general tool for obtaining optimal tradeoffs for \emph{all pairs} $(\delta,\Delta)$ that exhibit a set of elementary properties that are shared by the vast majority of distance measures that can be found in the literature. These properties, which are discussed in Sec.~\ref{sec:dist}, are essentially convexity and suitable forms of basis-(in)dependence.  For the case of a non-degenerate von Neumann target measurement Thm.~\ref{thm:universality} in Sec.~\ref{sec:vN1} shows that optimal devices can always be found within a universal two-parameter family, independent of the specific choice of $\delta$ and $\Delta$:

\begin{theorem*}[(Almost universal) optimal instruments] Let $\Delta$ and $\delta$ be distance-measures for quantifying disturbance and measurement-error that satisfy Assumptions~\ref{assum:1} and \ref{assum:2} (cf. Sec.~\ref{sec:dist}), respectively. Then the optimal $\Delta-\delta$-tradeoff w.r.t. a target measurement that is given by an orthonormal basis $\{|i\rangle\in\C^d\}_{i=1}^d$   is attained within the  two-parameter family of instruments defined by
\be\label{eq:optinst} I_i(\rho):= z\langle i|\rho|i\rangle\frac{\1_d-\ket{i}\bra{i}}{d-1}+(1-z)K_i\rho K_i,\quad K_i:=\mu\1_d+\nu\ket{i}\bra{i},
\ee where $z\in[0,1]$ and $\mu,\nu\in\R$ satisfy $d\mu^2+\nu^2+2\mu\nu=1$ (which makes $\sum_i I_i$ trace preserving).
\end{theorem*}
While the  parameter $z$ can be eliminated for instance in all cases mentioned above, we show in Cor.~\ref{cor:z} that this is not possible in general. 

If the target measurement itself is not a von Neumann measurement but a general POVM, then closed-form expressions like the ones above should not be expected. For the important case of the diamond norm, we show in Sec.~\ref{sec:SDP} how the optimal tradeoff curve can still be obtained via a semidefinite program (SDP). This is an instance of the following more general fact (Thm.~\ref{thm:SDPalg}):
\begin{theorem*}[SDP solution for arbitrary target measurements] If $\Delta$ and $\delta$ are both convex and semialgebraic, then the accessible region in the $\Delta-\delta
$-plane is the feasible set of a SDP.
\end{theorem*}
Note that no assumptions on the chosen measures are made other than being convex and semialgebraic.

%%%%%%%%%%%%%%%%%%%%%%%%%%%%%%%%%%%%%%%%
\section{Distance measures}\label{sec:dist}
%%%%%%%%%%%%%%%%%%%%%%%%%%%%%%%%%%%%%%%%
In this section we have a closer look at the functionals $\Delta:\cT_d\rightarrow [0,\infty]$ and $\delta:\cE_{d,m}\rightarrow [0,\infty]$ that quantify how much $E'$ and $T_1$ differ from $E$ and $\id$, respectively. We will not assume that they arise from metrics and use the notion of a `distance' merely in the colloquial sense.  We will state the assumptions that we will use in Sec.~\ref{sec:vN1} and discuss some of the most common measures that appear in the literature.\vspace*{5pt}
\paragraph{\bf Quantifying disturbance}
For the universality theorem (Thm.~\ref{thm:universality}) we will need the following assumption on $\Delta$:\footnote{In fact,  slightly less is required since Eq.~(\ref{eq:assumpDbi}) will only be used for unitaries that are products of diagonal and permutation matrices.}
\begin{assumption}[on the distance measure to the identity channel]\label{assum:1}\ \\
 For $ \Delta:\cT_d\rightarrow[0,\infty]$ we assume that (a)  $\Delta(\id)=0$, (b) $\Delta$ is convex, and (c) $\Delta$ is basis-independent in the sense that for every unitary $U\in\cM_d$ and every channel $\Phi\in\cT_d$:
  \be\Delta\Big(U\Phi(U^*\cdot U)U^*\Big)=\Delta(\Phi).\label{eq:assumpDbi}\ee
\end{assumption}
In the usually considered cases, $\Delta$ arises from a distance measure on the set of density operators $\cS_d\subseteq\cM_d$. In fact, if $\tilde{\Delta}:\cS_d\times\cS_d\rightarrow[0,\infty]$ is convex in its first argument, unitarily invariant and satisfies $\tilde{\Delta}(\rho,\rho)=0$, then considering the worst case as well as the average case w.r.t. the input state both lead to functionals that satisfy Assumption~\ref{assum:1}. More precisely, if $\mu$ is a unitarily invariant measure on $\cS_d$ and $S\subseteq\cS_d$ a unitarily closed subset (e.g., the set of all pure states), then the following two definitions can easily be seen to satisfy Assumption~\ref{assum:1}, see the appendix:
\begin{eqnarray*}
\Delta_\infty(\Phi)&:=& \sup_{\rho\in S} \tilde{\Delta}\big(\Phi(\rho),\rho\big),\\
\Delta_\mu(\Phi)&:=& \int_{\cS_d} \tilde{\Delta}\big(\Phi(\rho),\rho\big)\; \mathrm{d}\mu(\rho).
\end{eqnarray*}
While $\Delta_\infty$ quantifies the distance between $\Phi$ and $\id$ in the worst case in terms of $\tilde{\Delta}$, $\Delta_\mu$ does the same for the average case.
 
Concrete examples for $\tilde{\Delta}$ are (i) $\tilde{\Delta}(\rho,\sigma)=1-F(\rho,\sigma)$, where $F(\rho,\sigma):=||\sqrt{\rho}\sqrt{\sigma}||_1$ is the fidelity, (ii) the relative entropy and many other quantum $f$-divergences ~\cite{Hiai_Mosonyi_Petz_Beny_2011} including the Chernoff- and Hoeffding-distance and (iii) $\tilde{\Delta}(\rho,\sigma)=|||\rho-\sigma|||$, where $|||\cdot|||$ is any unitarily invariant norm such as the Schatten $p$-norms.

The latter can, in a similar vein, be used to define Schatten $p$-to-$q$ norm-distances to the identity channel 
$$\Phi\ \mapsto ||\Phi-\id||_{p\rightarrow q,n}:=\sup_{\rho\in\cS_{dn}}\frac{||(\Phi-\id)\otimes\id_n(\rho)||_q}{||\rho||_p},\quad q,p\in[1,\infty], n\in\mathbbm{N},$$ which also fulfill Assumption~\ref{assum:1}. Special cases are given by the \emph{diamond norm} $||\cdot||_\diamond:=||\cdot||_{1\rightarrow1,d}$, which we discuss in more  detail in Sec.~\ref{sec:gvN},  and its dual, the \emph{cb-norm} (with $p=q=\infty, n=d$).\vspace*{5pt}

\paragraph{\bf Quantifying measurement error}  The following assumptions that we need for the universality theorem on the functional $\delta$ refer to the case of a non-degenerate von Neumann target measurement that is given by an orthonormal basis $(|i\rangle\langle i|)_{i=1}^d$.
\begin{assumption}[on the distance measure to the target measurement]\label{assum:2}\ \\
 For $ \delta:\cE_d\rightarrow[0,\infty]$ we assume that (a)  $\delta\big((|i\rangle\langle i|)_{i=1}^d\big)=0$, (b) $\delta$ is convex, (c) $\delta$ is permutation-invariant in the sense that for every permutation $\pi\in S_d$ and any $M\in\cE_d$
 \be M_i'=U_\pi^* M_{\pi(i)} U_\pi\ \forall i\ \Rightarrow\ \delta(M')=\delta(M),\label{eq:assumpdperm}\ee
 where $U_\pi$ is the permutation matrix that acts as $U_\pi |i\rangle=|\pi(i)\rangle$,
 and (d)  that for every diagonal unitary $D\in\cM_d$ and any $M\in\cE_d$
 \be  M_i'= D^* M_i D \ \forall i\ \Rightarrow\ \delta(M')=\delta(M).\label{eq:assumpdessdiag}\ee
\end{assumption}
Here, the most common cases arise from distance measures $\tilde{\delta}:\cP_d\times\cP_d\rightarrow[0,\infty]$ on the space of probability distributions $\cP_d:=\big\{q\in\R^d|\sum_{i=1}^d q_i=1\wedge \forall i: q_i\geq 0\big\}$ applied to the target distribution $p_i:=\langle i|\rho|i\rangle$ and the actually measured distribution $p_i':=\tr{\rho E_i'}$. Suppose $\tilde{\delta}$ is convex in its second argument, invariant under joint permutations and satisfies $\tilde{\delta}(q,q)=0$. Then the worst-case as well as the average-case construction 
\bea
\delta_{\infty}(E')&:=&\sup_{\rho\in S} \tilde{\delta}(p,p'),\nonumber\\
\delta_{\mu}(E')&:=&\int_{\cS_d} \tilde{\delta}(p,p') \; \mathrm{d}\mu(\rho),\nonumber
\eea
 both satisfy Assumption~\ref{assum:2}, see appendix. Concrete examples for $\tilde{\delta}$ are all $l_p$-norms for $p\in[1,\infty]$ and the Kullback-Leibler divergence as well as  other $f$-divergences. Other examples for $\delta$ that satisfy Assumption~\ref{assum:2} are transport cost functions like the ones used in ~\cite{Schwonnek_Reeb_Werner_2016}. 

Note that convexity of the two measures $\Delta$ and $\delta$ implies that the region in the $\Delta-\delta$-plane that is accessible by quantum instruments is a convex set. The boundary of this set is given by two lines that are parallel to the axes (and correspond to the maximal values of $\Delta$ and $\delta$) and what we call the \emph{optimal tradeoff curve}.   
%%%%%%%%%%%%%%%%%%%%%%%%%%%%%%%%%%%%%%%%
\section{Universal optimal devices}\label{sec:vN1}
%%%%%%%%%%%%%%%%%%%%%%%%%%%%%%%%%%%%%%%%
There are three major steps towards proving the claimed universality theorem: the exploitation of symmetry, the construction of a von Neumann algebra isomorphism to obtain a manageable representation, and the final reduction to the envelope of a unit cone.

Throughout this section, the target measurement will be given by an orthonormal basis $E=(|i\rangle\langle i|)_{i=1}^d$.
In this case, instead of working with instruments it turns out to be slightly more convenient to work with channels. More specifically, we will describe the entire process by a channel $T:\cM_d\rightarrow\cM_d\otimes\cM_d$ with marginals $T_1,T_2\in\cT_d$. $T_1$ will then reflect the evolution of the `disturbed' quantum system, whereas the output of $T_2$ is measured by $E$ leading to $E_i'=T_2^*(E_i)$. This is clearly describable by an instrument and conversely, for every instrument $I$ we can simply construct
$$ T(\rho):=\sum_{i=1}^d I_i(\rho)\otimes |i\rangle\langle i|,$$
which shows that the two viewpoints are equivalent.

\begin{proposition}[Reduction to symmetric channels]\label{prop:sym} Let $G$ be the group generated by all diagonal unitaries and permutation matrices in $\cM_d$. If $\Delta$ and $\delta$ satisfy Assumptions~\ref{assum:1} and \ref{assum:2}, respectively, the optimal tradeoff between them can be attained within the set of channels $T:\cM_d\rightarrow\cM_d\otimes\cM_d$ for which 
\be (U\otimes U)T\big(U^*\rho U\big)(U\otimes U)^*\ =\ T(\rho)\quad \forall U\in G,\rho\in\cS_d.\label{eq:sym0}\ee
\end{proposition}
\begin{proof} We will show that for an arbitrary channel $T$, which does not necessarily satisfy Eq.~(\ref{eq:sym0}), the symmetrization
\be \bar{T}:=\int_G (U\otimes U)T\big(U^*\cdot U\big)(U\otimes U)^*\; \mathrm{d}U\nonumber\ee  w.r.t. the Haar measure of $G$ performs at least as well as $T$. Let $\bar{T}_1$ and $\bar{T}_2$ be the marginals of $\bar{T}$. Then
\begin{eqnarray*}
\Delta\big(\bar{T}_1\big)&=& \Delta\left(\int_G U T_1\big(U^*\cdot U\big) U^*\; \mathrm{d}U\right)\\
&\stackrel{(1b)}{\leq}&  \int_G \Delta\left( U T_1\big(U^*\cdot U\big) U^*\right)\; \mathrm{d}U\ \stackrel{(1c)}{=}\ \Delta(T_1),
\end{eqnarray*}
where the used assumption is indicated above the (in-)equality sign. Similarly, we obtain
\begin{eqnarray*}
\delta\left[\Big(\bar{T}_2^*\big(|i\rangle\langle i|\big)\Big)_{i=1}^d\right]&\stackrel{(2b)}{\leq}& \int_G \delta\left[\Big(U^* T_2^*\big(U|i\rangle\langle i| U^*\big)U\Big)_{i=1}^d\right]\; \mathrm{d}U\\
&\stackrel{(2d)}{=}&\int_G \delta\left[\Big(U_\pi^* T_2^*\big(|\pi(i)\rangle\langle\pi(i)|\big)U_\pi\Big)_{i=1}^d\right]\; \mathrm{d}U\\
&\stackrel{(2c)}{=}& \delta\left[\Big(T_2^*\big(|i\rangle\langle i|\big)\Big)_{i=1}^d\right],
\end{eqnarray*}
where we have used that every $U\in G$ can be written as $U=U_\pi D$, where $U_\pi$ is a permutation and $D$ a diagonal unitary, both depending on $U$. 

Consequently, when replacing $T$ by its symmetrization $\bar{T}$, which satisfies Eq.~(\ref{eq:sym0}) by construction, neither $\Delta$ nor $\delta$ is increasing.
\end{proof}

\begin{lemma}[Structure of marginals of symmetric channels]\label{lem:commutant} Let $G$ be the group generated by all diagonal unitaries and permutation matrices in $\cM_d$ and $\Phi:\cM_d\rightarrow\cM_d$ a quantum channel. Then the following are equivalent:
\begin{enumerate} \item 
$ \Phi(\rho)=U\Phi\big(U^*\rho U\big) U^*\quad \forall U\in G,\rho\in\cS_d$.
\item There are $\alpha,\beta,\gamma\in\R$ with $\alpha+\beta+\gamma=1$ so that 
\be\Phi=\alpha\;\tr{\cdot}\frac{\1}{d}+\beta\;\id+\gamma\sum_{i=1}^d |i\rangle\langle i|\langle i|\cdot|i\rangle.\label{eq:commutant}\ee
\end{enumerate} 
\end{lemma}
\begin{proof} (2) $\Rightarrow$ (1) can be seen by direct inspection. In order to prove the converse, we consider the Jamiolkowski-state (= normalized Choi-matrix) $J_\Phi:=\frac1d \sum_{i,j=1}^d \Phi\big(|i\rangle\langle j|\big)\otimes |i\rangle\langle j|$. Then (1) is equivalent to the statement that $J_\Phi$ commutes with all unitaries of the form $U\otimes\bar{U}$, $U\in G$. Considering for the moment only the subgroup of diagonal unitaries, this requires that
$$ \langle ij|J_\Phi|kl\rangle =(2\pi)^{-d}\int_0^{2\pi} \ldots \int_0^{2\pi} e^{i(\varphi_i-\varphi_j-\varphi_k+\varphi_l)}\langle ij|J_\Phi|kl\rangle\; \mathrm{d}\varphi_1\ldots d\varphi_d,
$$ which vanishes unless $(i=j\wedge k=l)\vee(i=k\wedge j=l)$. Hence, there are $A,B\in\cM_d$ such that 
$$J_\Phi=\sum_{i,j=1}^d A_{ij}|i \rangle \langle i|\otimes |j\rangle \langle j| + B_{ij}|i\rangle\langle j| \otimes |i\rangle\langle j|.
$$
Next, we will exploit that $J_\Phi$  commutes in addition with permutations of the form $U_\pi\otimes U_\pi$ for all $\pi\in S_d$. For $i\neq j$ this implies that $A_{i,j}=A_{\pi(i),\pi(j)}$ and $B_{i,j}=B_{\pi(i),\pi(j)}$ so that there is only one independent off-diagonal element for each $A$ and $B$. The case $i=j$ leads to a third parameter that is a coefficient in front of $\sum_i|ii\rangle\langle ii|$. Translating this back to the level of quantum channels then yields Eq.~(\ref{eq:commutant}). The coefficients are real and sum up to one since $\Phi$ preserves hermiticity as well as the trace.
\end{proof}
If $T$ is symmetric as in Prop.~\ref{prop:sym}, then both marginal channels $T_1$ and $T_2$ are of the form derived in the previous Lemma. That is, each $T_i$, $i\in\{1,2\}$, is specified by three parameters $\alpha_i,\beta_i,\gamma_i$ only two of which are independent.

The following Lemma shows that under Assumption~\ref{assum:2}  the error measure $\delta$ depends only on $\alpha_2$ and does so in a non-decreasing way.  
\begin{lemma}\label{lem:delta=a2} Let $\delta$ satisfy Assumption~\ref{assum:2}. There is a non-decreasing function $\hat{\delta}:[0,1]\rightarrow[0,\infty]$ s.t. for all $T_2:\cM_d\rightarrow\cM_d$ of the form in Eq.~(\ref{eq:commutant}) with coefficients $\alpha_2,\beta_2,\gamma_2$ we have  $\delta \big[ \big(T_2^*(|i\rangle\langle i|)\big)_{i=1}^d\big] =\hat{\delta}(\alpha_2)$. 
\end{lemma}
\begin{proof}
The statement  follows from convexity of $\delta$ together with the observation that $\beta$ and $\gamma$ only contribute jointly to $\delta$ and not individually. This is seen by composing $T_2$ with the projection onto the diagonal. This leads to a channel of the same form, but possibly different parameters. On the level of the latter  the composition corresponds to  $(\alpha_2,\beta_2,\gamma_2)\mapsto(\alpha_2,0,\beta_2+\gamma_2)$. The distance measure $\delta$, however, does not change in this process and thus depends only on the sum $\beta_2+\gamma_2$ and not on those two parameters individually. As this sum equals $1-\alpha_2$ we see that $\delta$ can be regarded as a function of $\alpha_2$ only. We formally denote this function by $\hat{\delta}$.  Assumption (2b) then implies that $\hat{\delta}$ is convex. As it is in addition positive and satisfies $\hat{\delta}(0)=0$ by Assumption (2a), we get that $\hat{\delta}$ is non-decreasing. 
\end{proof}

For later investigation, it is useful to decompose the $J_\Phi$ that corresponds to Eq.~(\ref{eq:commutant}) into its spectral projections:
\begin{eqnarray}
J_\Phi &=& aP_a+bP_b+cP_c,\quad\text{where}\quad P_a:=\1-\sum_{i=1}^d |ii\rangle\langle ii| ,\nonumber\\
& &P_b:= \frac1d\sum_{i,j=1}^d |ii\rangle\langle jj|,\label{eq:specproj}\quad P_c:= \sum_{i=1}^d  |ii\rangle\langle ii| -P_b.
\end{eqnarray} The coefficients $a,b,c$ are the eigenvalues of $J_\Phi$ (and thus non-negative) and related to $\alpha,\beta,\gamma$ via $\alpha=d^2 a,\ \beta=b-c,\ \gamma=d(c-a)$. When considering symmetric $T$, we will label the eigenvalues of $J_{T_i}$ with a subscript $i\in\{1,2\}$ to distinguish the two marginals. 

Since the $P$'s are mutually orthogonal projectors, we can obtain the eigenvalues from their expectation values. That is,
\be x_1=\frac{\tr{(P_x\otimes\1) J_T}}{\tr{P_x}}\quad\text{and}\quad x_2=\frac{\tr{(\1\otimes P_x) J_T}}{\tr{P_x}},\quad x\in\{a,b,c\}.\label{eq:eigexp}\ee
If we are aiming at identifying  a subset of optimal channels, we can, according to Lemma~\ref{lem:delta=a2}, w.l.o.g. use  $a_2$ as $\delta$. Due to the monotonic  relation between the two, optimality for one implies optimality for the other. The question we are going to address in the next step of the argumentation is then: which values of $a_1, b_1$ and $c_1$  are consistent with a given value of $a_2$? After all, due to Prop.~\ref{prop:sym}, $\Delta$ and $\delta$ will be functions of those parameters only. Thus, we would like to know which is the accessible region in the space of these parameters, when we vary $J_T$ over the set of all density matrices. 

We tackle this question using an operator algebraic point of view: the operators $\1\otimes P_a,P_x\otimes\1$ together with the identity operator generate a von Neumann algebra $\cA$ on which $J_T$ acts as a state, i.e., as a normalized positive linear functional. This suggests the use of a von Neumann algebra isomorphism that simplifies the representation. To this end, we observe that $\cA$ is generated by the following operators:
\begin{align*}
\1_{d^3} &=:\ \Ga & \1_d\otimes\sum_{i=1}^d |ii\rangle\langle ii| &=:\ \Gb\\
\sum_{i,j=1}^d |ii\rangle\langle jj|\otimes\1_d &=:\ \Gc & \sum_{i=1}^d |ii\rangle\langle ii|\otimes\1_d &=:\ \Gd
\end{align*}
The introduced diagrammatic notation turns out be useful as it reflects that these operators are what one may call \emph{contraction tensors}.\footnote{Please note that these diagrams are not braid diagrams, but rather diagrammatically represent contraction tensors.} If we view an element in $\cM_d\otimes\cM_d\otimes\cM_d$ as a tensor with three left and three right indices, then the diagrammatic notation indicates which of these indices get contracted together---by connecting them.
Taking products of pairs of these four operators generates (up to scalar multiples, which arise from closed loops) three new contraction tensors:
\begin{equation*}
\Ge := \Gb\Gc,\quad \Gf := \Gc\Gb,\quad \Gg :=\Gb\Gd.
\end{equation*}
The set of these seven tensors is, however, closed under multiplication (again ignoring scalar multiples). This is easily verified by using the diagrammatic notation and going through all cases. This observation is the core for constructing a simplifying isomorphism:
\begin{lemma}[Isomorphic representation]\label{lem:iso}
Let $\cA$ be the von Neumann algebra that is generated by the set $\{\1_{d^3},\1_d\otimes P_a,P_a\otimes\1_d ,P_b\otimes\1_d, P_c\otimes\1_d\}$. A unital map  $\iota:\cA\rightarrow\cM_2\oplus\C^3$ defined by
\begin{align}
\iota: \Gc &\mapsto  d|e_1\rangle\langle e_1|& \iota: \Gg &\mapsto  |e_2\rangle\langle e_2|\label{eq:iso1}\\
\iota: \Gd &\mapsto \1_2 \oplus f_2& \iota: \Gb &\mapsto |e_2\rangle\langle e_2| \oplus f_1\label{eq:iso2}
\end{align}
is an isomorphism if $|e_1\rangle, |e_2\rangle$ constitute unit vectors with $|\langle e_1|e_2\rangle|^2=1/{d}$ in the space of the non-abelian part (i.e., the corresponding projections as well as $\1_2$ are in $\cM_2$) and $f_1:=(1,0,0),f_2:=(0,1,0)$ are elements of the abelian part.\footnote{Here we regard $\C^3$ as space $\cM_1\oplus\cM_1\oplus\cM_1$ of diagonal matrices in $\cM_3$.}
\end{lemma}
\begin{proof}
$\cA$ is generated by the above set of seven contraction tensors. Since this set is closed under multiplication, $^*$-operation and contains linear independent elements, we have ${\rm dim}(\cA)=7$. Moreover, $\cA$ is non-commutative since $[\Gc,\Gg]\neq 0$. From the representation theory of finite-dimensional von Neumann algebras we known that every $7$-dimensional non-commutative von Neumann algebra is isomorphic to $\cM_2\oplus\C^3$ \cite[Thm. 5.6]{Farenick_2001}. Hence, we can establish an isomorphism $\iota$ by representing a generating set of $\cA$ in $\cM_2\oplus\C^3$. Due to unitality $\iota(\1_{d^3})=\1_2\oplus(1,1,1)$ has to hold. Moreover, since $\Gc,\Gg$ are (proportional to) non-commuting minimal projectors in $\cA$, they need to be the same in $\cM_2\oplus\C^3$. Taking proportionality factors into account, this determines Eq.~(\ref{eq:iso1}) and requires $|\langle e_1|e_2\rangle|^2=1/{d}$ in order to be consistent with the value of the trace $\tr{\Gc\Gg}$. From $\Gd\Gc=\Gc$ and $\Gd\Gg=\Gg$ we see that $\iota(\Gd)$ acts as identity on $\cM_2$. Similarly, $\iota(\Gb)$, when restricted to $\cM_2$, has to be a projector that is not the identity and has $|e_2\rangle$ as eigenvector (due to $\Gb\Gg=\Gg$). This determines Eq.~(\ref{eq:iso2}) when restricted to $\cM_2$. Moreover, since $\cM_2\oplus\C^3$ has to be generated, both $\iota(\Gd)$ and $\iota(\Gb)$ have to have non-zero  parts on the abelian side. Since they are projectors, these parts need to be projectors as well. Finally, they have to be one-dimensional since otherwise the identity operator would become linearly dependent. 
\end{proof}
Using this Lemma we can now express the accessible region within the space of parameters  $\alpha_1,\beta_1,\gamma_1,\alpha_2$ by varying over all states on $\cM_2\oplus\C^3$, instead of over all states $J_T$ on $\cM_{d^3}$. To this end, we just have to unravel the linear maps from the parameters to the eigenvalues $a_1,b_1,c_1,a_2$, to the $P$'s, to the contraction tensors, and finally to their representation in $\cM_2\oplus\C^3$. In this way, we obtain:
\begin{corollary}
\label{cor:reduc}
There exists a channel $T:\cM_d\rightarrow\cM_d\otimes\cM_d$ with corresponding Jamiolkowski state $J_T$ whose marginals give rise to the parameters $\alpha_1,\beta_1$ and $a_2$ iff there exists a state $\varrho$ on $\cM_2\oplus\C^3$ such that
\bea
\alpha_1 &=& \frac{d}{d-1}\Big(1-\tr{\1_2\varrho}-\tr{f_2\varrho}\Big),\label{eq:paramred1}\\
\beta_1 &=& \bra{e_1}\varrho\ket{e_1}-\frac{1}{d-1}\Big(\tr{\1_2\varrho}+\tr{f_2\varrho}-\bra{e_1}\varrho\ket{e_1}\Big),\label{eq:paramred2}\\
a_2 &=& \big(1- \bra{e_2}\varrho\ket{e_2}-\tr{f_1\varrho}\big)/(d^2-d),\label{eq:paramred3}
\eea where $\C^3$ is regarded as space of diagonal $3\times 3$ matrices and $e_1,e_2,f_1,f_2$ are as in Lemma~\ref{lem:iso}.
\end{corollary}
The proof of this corollary can be found in the appendix.

There is still unitary freedom in the choice of the vectors $e_1, e_2$. We utilize this and set
\be \langle e_1|\sigma_y|e_1\rangle=\langle e_2|\sigma_y|e_2\rangle=0\quad\text{and}\quad |e_2\rangle\langle e_2|=\frac12(\1_2+\sigma_x)\label{eq:fixinplane},\ee where the $\sigma_i$'s are the usual Pauli matrices. So in particular, we choose the vectors such that the corresponding projectors lie in an equatorial plane of the Bloch sphere that is characterized by density matrices with real entries.

In order to simplify the problem further, we now focus more explicitly on minimizing $a_2$: 
\begin{proposition}[Reduction to the unit cone]\label{prop:cone} Under the constraints given by Eqs.~(\ref{eq:paramred1} -- \ref{eq:fixinplane}), the minimum value for $a_2$ for arbitrary fixed values of $\alpha_1,\beta_1$ that is achievable by varying over all states $\varrho$ is attained for a state of the form 
\be\varrho=\frac{1}{2}\Big((1-z)\1_2+x\sigma_x+y\sigma_z\Big)\oplus (z,0,0),\label{eq:cone}\ee
where $(x,y,z)\in\R^3$ is an element of the envelope of the unit cone, i.e., $z\in[0,1], x^2+y^2=(1-z)^2$.
\end{proposition}
\begin{proof}
We simplify the structure of $\varrho$ in four steps, each of which eliminates one parameter. First, note that we can assume $\tr{f_3\varrho}=0$, where $f_3$ is the diagonal matrix $(0,0,1)$. This is seen by considering the map $\varrho\mapsto\varrho+\tr{f_3\varrho}(f_1-f_3)$, which decreases $a_2$, sets the $f_3$-component to zero, but leaves $\alpha_1$ and $\beta_1$ unchanged. 

Second, we claim that the $f_2$-component can be set to zero, as well. To this end, consider the map $\varrho\mapsto\varrho+\tr{f_2\varrho}(|e_1^\perp\rangle\langle e_1^\perp|-f_2)$ where $e_1^\perp$ is a unit vector in $\C^2$ that is orthogonal to $e_1$. By construction, this sets the $f_2$-component to zero, decreases $a_2$ and leaves $\alpha_1$ and $\beta_1$ invariant. Taken together with the first step, this already shows that the abelian part of $\varrho$ can be assumed to be of the form $(z,0,0)$ for some $z\in[0,1]$. 

Third, observe that the $\sigma_y$-component of the non-abelian part of $\varrho$ does not enter any of the equations so that we can as well set it to zero and thus assume that, restricted to $\cM_2$, $\varrho$ lies in the 'real' equatorial plane of the Bloch sphere. 

Taking positivity and normalization into account, Eq.~(\ref{eq:cone}) summarizes these findings, so far with $x^2+y^2\leq (1-z)^2$. What remains to show is that equality can be assumed, here. Let $v_1,v_2,w\in{\R^3}$ be the Bloch vectors of $e_1$, $e_2$ and $\varrho$, respectively. Suppose $||w||_2<1$, which corresponds to a point that does not lie on the envelope of the cone and let $v_1^\perp\in\R^3$ be a unit vector in the equatorial plane that is orthogonal to $v_1$. Then the map $w\mapsto w+\epsilon v_1^\perp$, for sufficiently small $\epsilon$ of the right sign, leaves $\alpha_1$ and $\beta_2$ unchanged, but decreases $a_2$. Hence, we can choose $\epsilon$ so that the Bloch vector reaches unit norm, which completes the proof of the proposition. 
\end{proof}
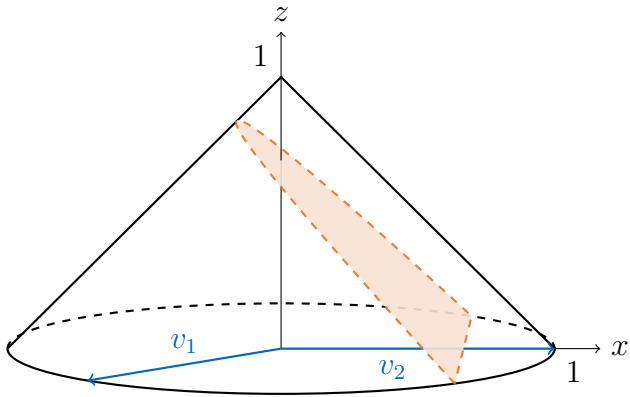
\begin{figure}[ht!]
\centering
\begin{tikzpicture}[scale=0.6]
\draw[->] (0,0) -- (7,0)node[right]{$x$};
\draw[<-, name path=y axis] (0,7)node[above]{$z$} -- (0,0);
\draw (6,1pt) -- (6,-1pt) node[below right]{$1$};
\draw (1pt,6) -- (-1pt,6) node[above left]{$1$};
\draw[name path=hat, thick] (-6,0) -- (0,6) -- (6,0);
\draw[dashed, name path=upper arc, thick] (-6,0) arc (180:0:6cm and 1cm);
\draw[name path=lower arc, thick] (-6,0) arc (180:360:6cm and 1cm);
\draw[tumblue, thick,->] (0,0) -- node[below left]{$v_2$} (6,0);
%v1
\path[name path=P1] (0,0) -- (-6,-1); %Hilflinie P1
\path[name intersections={of=P1 and lower arc, by=A1}]; %Intersection von Arc mit Hilfslinie 
\draw[tumblue, thick, ->] (0,0) -- node[above]{$v_1$} (A1);
%parabola
\path[name path=P2] (4,0) -- (-2,6); %Hilfslinie P2
\draw[name intersections={of=P2 and hat, by=A2}]; 
  %\filldraw[fill=black!20!white] (A2) circle (0.1cm)  let \p2=(A2) in node[above]{$A2(\x2,\y2)$};
%
\path[name path=P3] (4,0) -- (3.75,-1); %Hilfslinie P3
\draw[name intersections={of=P3 and lower arc, by=A3}]; 
	%\filldraw[fill=black!20!white] (A3) circle (0.1cm) let \p3=(A3) in node[above]{$A3(\x3,\y3)$};
\path[name path=P4] (5,4) -- (4,0); %Hilfslinie P4
\draw[name intersections={of=P4 and upper arc, by=A4}];
	%\filldraw[fill=black!20!white] (A4) circle (0.1cm) let \p4=(A4) in node[above]{$A4(\x4,\y4)$};
\filldraw[fill=tumorange!20!white, draw= tumorange, opacity=0.9, dashed, name path=para, thick] (A3) .. controls (-2.7,6.4) and (-2.6,7.0) .. (A4) -- (A3); %control points were calculated brute force
\path[name intersections={of=para and y axis}];
\coordinate (I1) at (intersection-1);
\coordinate (I2) at (intersection-2);
	%\filldraw[fill=black!20!white] (I1) circle (0.1cm)  node[above]{$I1$};
	%\filldraw[fill=black!20!white] (I2) circle (0.1cm)  node[above]{$I2$};
\node (I3) at ($(I1)!0.4!(I2)$) {};
	%\filldraw[fill=black!20!white] (I3) circle (0.1cm)  node[above]{$I3$};
\draw (I3) -- (I2);
\end{tikzpicture}
 \caption{Sketch of the unit cone used in the construction of the proof in Prop.~\ref{prop:cone}. The orange parabola corresponds to a fixed value of $\delta$ and the optimal device is contained within its boundary; its location depends on the chosen disturbance distance measure $\Delta$.}%
 \label{fig:plot_cone}%
 \end{figure}
This completes the list of ingredients that are needed for the main theorem of this section:
\begin{theorem}[(Almost universal) optimal devices]\label{thm:universality} Let $\Delta$ and $\delta$ be distance-measures for quantifying disturbance and measurement-error that satisfy Assumptions~\ref{assum:1} and \ref{assum:2}, respectively. Then the optimal $\Delta-\delta$-tradeoff is attained within the following two-parameter family of quantum channels:
\bea\label{eq:optchannels} T(\rho)&:=& \sum_{i=1}^d \left[z\langle i|\rho|i\rangle\frac{\1_d-\ket{i}\bra{i}}{d-1}+(1-z)K_i\rho K_i\right]\otimes\ket{i}\bra{i}, \\
&& K_i:=\mu\1_d+\nu\ket{i}\bra{i}, \nonumber
\eea where $z\in[0,1]$ and $\mu,\nu\in\R$ are constrained by imposing $T$ to be trace preserving.
\end{theorem}
\begin{proof}
What remains to do is to translate the two-parameter family of Eq.~(\ref{eq:cone}) into the world of channels. It suffices to consider the  cases in which either $z=0$ or $z=1$ since these generate the general case by convex combination. In both cases the relevant von Neumann algebra is a factor on which the dual of $\iota$ becomes its inverse, up to a multiplicity factor. This means, we have to compute $\iota^{-1}(\varrho)$ and show that it equals $J_T$  when normalized. 

If $z=1$, this is readily verified since in this case $\varrho=f_1$ for which Eqs.~(\ref{eq:iso1},\ref{eq:iso2}) give
$$\iota^{-1}(\varrho)=\Gb-\Gg =\sum_{i=1}^d \big(\1_d-|i\rangle\langle i|\big)\otimes |ii\rangle\langle ii|.$$
If $z=0$ then $\varrho$ is a rank-one projection within the real algebra generated by the projections onto $\ket{e_1}$ and $\ket{e_2}$. That is, 
\be\varrho = \mu^2 \ket{e_1}\bra{e_1}+\frac{\nu^2}{d} \ket{e_2}\bra{e_2} +\tau \big(\ket{e_1}\bra{e_2}+\ket{e_2}\bra{e_1}\big),\nonumber \ee
for some $\tau,\mu, \nu\in\R$. Having rank one requires vanishing determinant, which fixes $\tau^2=\mu^2\nu^2/d$ while the remaining two parameters are constrained by the normalization  $\tr{\varrho}=1$. Please note that we choose $\tau=\mu\nu/\sqrt{d}$, since $\mu \in \R$, which thus includes the other case. Exploiting that $\iota^{-1}$ is again an isomorphism and that for instance $\ket{e_1}\bra{e_2}=\sqrt{d}\ket{e_1}\bra{e_1}\cdot\ket{e_2}\bra{e_2}$, we obtain
\bea
\iota^{-1}(\varrho)&=&\frac{1}{d}\left[\mu^2\;\Gc+\nu^2\;\Gg+\mu\nu\big(\Ge+\Gf\big)\right]\nonumber\\
&=& \frac{1}{d}\sum_{i,k,l=1}^d  K_i\ket{k}\bra{l}K_i\otimes \ket{k}\bra{l}\otimes\ket{i}\bra{i},\nonumber
\eea which is, up to normalization, indeed the Choi matrix of the claimed channel.
\end{proof}

In the following section we will see that for many common disturbance measures $\Delta$, in fact, one more parameter can be eliminated: $z=0$ turns out to be optimal if $\Delta$ is for instance constructed from the average-case or worst-case fidelity, the worst-case Schatten $1-1$-norm or the diamond norm. This may not come as a surprise since a look at Eq.~(\ref{eq:paramred1}) reveals that for channels that correspond to elements of the unit cone we have 
\be\label{eq:alpha1z} \alpha_1=\frac{d}{d-1} z. \ee
In other words, the contribution of the completely depolarizing channel to $T_1$ vanishes iff $z=0$.
 This raises the question whether $z=0$ is generally optimal under Assumptions~\ref{assum:1} and \ref{assum:2}. The following construction, whose only purpose is to enable the argument, shows that this is not true. Hence, without adding further assumptions about the distance measures (in particular about $\Delta$) no further reduction is possible. On the set of quantum channels on $\cM_d$ we define 
$$\DD(\Phi):=\sup_{||\psi||=1}\bra{\psi}\Phi\big(\ket{\psi}\bra{\psi}\big)\ket{\psi}-\inf_{||\varphi||=1}\bra{\varphi}\Phi\big(\ket{\varphi}\bra{\varphi}\big)\ket{\varphi}.$$
This particular example yields zero disturbance for the depolarizing channel, and thus allows to show that $z=0$ is not true in general.
\begin{lemma} $\DD$ satisfies Assumption~\ref{assum:1}.
\end{lemma}
\begin{proof}
Evidently, $\DD(\id)=0$ and $\DD$ is basis-independent. Convexity follows from the fact that $\DD$ is a supremum over linear functionals.
\end{proof}
\begin{corollary}[Necessity of the second parameter]\label{cor:z}
Let $\delta $ be any error-measure that satisfies Assumption~\ref{assum:2} and that is faithful in the sense that $\delta=0$ implies a perfect measurement. Then the optimal $\DD-\delta$-tradeoff cannot be attained within the family of channels in Eq.~(\ref{eq:optchannels}) with $z=0$. 
\end{corollary}
\begin{proof}
Consider $\delta=0$ in the $\DD-\delta$-plane. Within the full set of channels in Eq.~(\ref{eq:optchannels}) there is one  that attains $\delta=0$ while $T_1(\cdot)=\tr{\cdot}\1/d$, by choosing $\mu =0$, $\nu = 1$ and $z= (d-1)/d$. The latter implies $\DD(T_1)=0$.
However, if we restrict ourselves to channels with $z=0$, then the unique channel in Eq.~(\ref{eq:optchannels}) that achieves $\delta=0$ has $T_1(\cdot)=\sum_i \bra{i}\cdot\ket{i}\ket{i}\bra{i}$ for which clearly $\DD(T_1)>0$.
\end{proof}
Clearly, $\DD$ is not a 'natural' disturbance measure. For instance, it has the somewhat odd property that it vanishes for the ideal channel as well as for the projection onto the maximally mixed state. In particular, it is not faithful. Note, however, that adding the latter as an additional requirement to Assumption~\ref{assum:1}, would still not allow to eliminate the parameter $z$. In order to construct a new counterexample, we could just consider $\Phi\mapsto\DD(\Phi)+\epsilon||\Phi-\id||_\diamond$. This would be faithful and satisfy Assumption~\ref{assum:1} for any $\epsilon > 0$, but for sufficiently small $\epsilon$, the minimum $\Delta$-value for $\delta=0$ would, by continuity, again not be attainable for $z=0$.

%%%%%%%%%%%%%%%%%%%%%%%%%%%%%%%%%%%%%%%%
\section{Optimal tradeoffs}\label{sec:vN2}
%%%%%%%%%%%%%%%%%%%%%%%%%%%%%%%%%%%%%%%%

In this section we will continue considering non-degenerate von Neumann measurements and exploit the universality theorem of the previous section in order to explicitly compute the optimal tradeoff for a variety of worst-case distance measures. We first discuss the total variational distance as a paradigm for the measurement error $\delta$ and then the fidelity and trace-norm as means for quantifying disturbance.

%%%%%%%%%%%%%%%%%%%%%%%%%%%
\subsection{Total variation}
We saw in Lemma~\ref{lem:delta=a2} that all functionals quantifying the measurement error consistent with Assumption~\ref{assum:2} are non-decreasing functions of the parameter $\alpha_2$. In the following, we want to make this dependence explicit for one case that we regard as the most important one from an operational point of view --- the worst-case total variational distance. Given two finite probability distributions $p$ and $p'$, their total variational distance is given by 
\be||p-p'||_{TV}:=\frac12||p-p'||_1=\frac12\sum_{i}|p_i-p_i'|.\ee
The significance of this distance stems from the fact that it displays the largest possible difference in probabilities that the two distributions assign to the same event. In our context the two probability distributions arise from an ideal and an approximate measurement on a quantum state. As  $||p-p'||_{TV}$ has itself a 'worst-case interpretation' it is natural to also consider the worst case w.r.t. all quantum states and use the resulting functional as $\delta$. That is,
\be \delta_{TV}\left(E'\right)=\sup_{\rho}\frac12\sum_i\big|\tr{E_i'\rho}-\langle i|\rho|i\rangle\big|.\label{eq:Eprimea2}\ee If $E_i'=T_2^*(|i\rangle\langle i|)$ with $T_2$ of the form in Eq.~(\ref{eq:commutant}) so that we can regard $\delta_{TV}$ as a function of $\alpha_2$, we will write $\hat{\delta}_{TV}(\alpha_2)$. 
\begin{lemma}[Total variational distance]\label{lem:TV}
In the symmetric setting discussed above, the worst-case total variational distance, regarded as a function of $\alpha_2$, is given by $\hat{\delta}_{TV}(\alpha_2)=\alpha_2(1-1/d)$. Furthermore, if an instrument is parametrized by the unit cone coordinates of Eq.~(\ref{eq:cone}), then it leads to a worst-case total variational distance of $(1-z-x)/2$.
\end{lemma} 
\begin{proof}
Inserting  $E_i'=T_2^*(|i\rangle\langle i|)=\alpha_2\1/d+(1-\alpha_2)|i\rangle\langle i|$ into Eq.~(\ref{eq:Eprimea2}) we obtain
\bea
\hat{\delta}(\alpha_2) &=& \alpha_2\;\sup_\rho  \frac12 \sum_i\left|\tr{\rho\big(\1/d-|i\rangle\langle i|\big)}\right| \nonumber\\
&=& \alpha_2\left(1-\frac{1}{d}\right),\nonumber
\eea
where the supremum is computed by first realizing that diagonal $\rho$'s (i.e., classical probability distributions) suffice and then noting that convexity of the $l_1$-norm allows to restrict to the extreme points of the simplex of classical distributions, which all lead to the same, stated value. 

The $\delta_{TV}$-value of an instrument parametrized by the coordinates of the unit cone can then be obtained from Eq.~(\ref{eq:paramred3}) when using that $\alpha_2=d^2 a_2$.
\end{proof}
An alternative way of quantifying the measurement error would be the worst-case $l_\infty$-distance between the two probability distributions $p$ and $p'$. In the present context, this measure turns out to have exactly the same value since
\bea
\sup_{\rho}\max_i \Big|\tr{E_i'\rho}- \langle i|\rho|i\rangle\Big| &=&\nonumber \max_i \big|\big|E_i'-|i\rangle\langle i|\big|\big|_\infty\\
&=&\alpha_2\big|\big|\1/d-|i\rangle\langle i|\big|\big|_\infty\;=\;\alpha_2\left(1-\frac1d\right).\nonumber
\eea 

%%%%%%%%%%%%%%%%%%%%%
\subsection{Worst-case fidelity}
We consider the worst-case fidelity of a channel $T_1:\cM_d\rightarrow\cM_d$ 
\be f:=\inf_{||\psi||=1}\langle\psi |T_1\big(|\psi\rangle\langle\psi|\big)|\psi\rangle,\label{eq:worstfT1}\ee
which is equal to $\inf_\rho F\big(T_1(\rho),\rho\big)^2$ due to joint concavity of the fidelity. The following 
states the optimal 'information-disturbance tradeoff' between $f$ and the total variational distance:
\begin{theorem}[Total variation - fidelity tradeoff]\label{thm:TVfidelity}
Consider a non-degenerate von Neumann measurement, given by an orthonormal basis in $\C^d$, and an instrument with $d$ corresponding outcomes. Then the worst-case total variational distance $\delta_{TV}$ and the worst-case fidelity $f$ satisfy
\be \delta_{TV}\geq \left\{\begin{array}{ll}\frac{1}{d}\left|\sqrt{f(d-1)}-\sqrt{1-f} \right|^2&\ \text{if }\ f\geq\frac{1}{d},\\ 0 &\ \text{if }\ f\leq \frac{1}{d}.
\end{array}\right.\label{eq:optfidtv}\ee
The inequality is tight and equality is attainable within the one-parameter family of instruments in Eq.~(\ref{eq:optinst}) with $z=0$.
\end{theorem}
\begin{proof}
We exploit that the optimal tradeoff is attainable for symmetric channels (Prop.~\ref{prop:sym}) whose marginal is given in Eq.~(\ref{eq:commutant}). Inserting this into the worst-case fidelity in Eq.~(\ref{eq:worstfT1}) we obtain
\bea
f &=&\min_{||\psi||=1}\left(\frac{\alpha_1}{d}+\beta_1+\gamma_1\sum_{i=1}^d |\langle\psi|i\rangle|^4\right) \nonumber\\
&=&\frac{\alpha_1}{d}+\beta_1+\left\{\begin{array}{ll}
\frac{\gamma_1}{d}&\ \text{if }\ \gamma_1\geq 0,\\
\gamma_1&\ \text{if }\ \gamma_1 <0.
\end{array}\right.\label{eq:fgamma}
\eea
Using Eqs.~(\ref{eq:paramred1},\ref{eq:paramred2}) together with $\gamma_1=1-\alpha_1-\beta_1$ we can express this in terms of the state $\varrho$. From the proof of Prop.~\ref{prop:cone} we know in addition that we can w.l.o.g. assume that $\tr{\varrho f_2}=0$ and $\tr{\1_2\varrho}=1-\tr{\varrho f_1}$. In this way, we obtain
\be f=\min\{1-\tr{\varrho f_1},\langle e_1|\varrho| e_1\rangle+\tr{\varrho f_1}/d\}.\label{eq:fminvarrho}
\ee We aim at maximizing Eq.~(\ref{eq:fminvarrho}) for each value of the total variational distance, which by Lemma~\ref{lem:TV} and Eq.~(\ref{eq:paramred3}) can be expressed as
\be \delta_{TV}=1-\langle e_2|\varrho| e_2\rangle-\tr{\varrho f_1}.\nonumber\ee
Considering the map $\varrho\mapsto\varrho+\epsilon|e_2\rangle\langle e_2|-\epsilon f_1$, $\epsilon\geq 0$, under which $\delta_{TV}$ is constant and $f$ non-decreasing, we see that $\tr{\varrho f_1}=0$ can be assumed. That is, $z=0$ is indeed sufficient for the optimal tradeoff.

The remaining optimization problem can be solved in the equatorial plane of the Bloch sphere, where $\varrho,|e_2\rangle\langle e_2|$ and $|e_1\rangle\langle e_1|$ are represented by Bloch vectors $(x,y)=:w,(1,0)$ and $(2/d-1,2\sqrt{d-1}/d)=:v$, respectively.
Minimizing $\delta_{TV}=(1-x)/2$ under the constraints $$ f\leq \frac12\big(1+\langle w,v\rangle\big),\quad \langle w, w\rangle=1,$$
then amounts to a quadratic problem whose solution is stated in Eq.~(\ref{eq:optfidtv}).
\end{proof}

%%%%%%%%%%%%%%%%%%%%%
\subsection{Average-case fidelity}
One prominent example of an average-case measure is the average-case fidelity of a quantum channel $T_1: \cM_d \to \cM_d$
\begin{equation}
\bar{f}:= \int_{\norm{\psi}=1} \bra{\psi} T_1(\kb{\psi}{\psi}) \ket{\psi}  \; \mathrm{d}\psi.
\label{eq:Avf}
\end{equation}
The following theorem gives the optimal 'information-disturbance tradeoff' between the average-case fidelity and the worst-case total variational distance:
\begin{theorem}[Total variation - average fidelity tradeoff]\label{thm:TVAvfidelity}
Consider a non-degenerate von Neumann measurement, given by an orthonormal basis in $\C^d$, and an instrument with $d$ corresponding outcomes. Then the worst-case total variational distance $\delta_{TV}$ and the average-case fidelity $\bar{f}$ satisfy 
\be 
\delta_{TV}\geq \left\{\begin{array}{ll}\frac{1}{d}\left|\sqrt{\left( \bar{f} - \frac{1}{d+1}\right)\frac{d^2-1}{d}}-\sqrt{\left(1- \bar{f}\right)\frac{d+1}{d}} \right|^2&\ \text{if }\ \bar{f}\geq\frac{2}{d+1},\\ 0 &\ \text{if }\ \bar{f}\leq \frac{2}{d+1}.
\end{array}\right.\label{eq:optavfidtv}\ee
The inequality is tight and equality is attainable within the one-parameter family of instruments in Eq.~(\ref{eq:optinst}) with $z=0$.
\end{theorem}

\begin{proof}
We again use the fact that the optimal tradeoff is attainable for symmetric channels by Prop.~\ref{prop:sym} and its marginal is given in Eq.~(\ref{eq:commutant}).  The average-case fidelity given in Eq.~(\ref{eq:Avf}) therefore yields
\bea
\bar{f} &=& \int_{\norm{\psi}=1} \bra{\psi} \left(\alpha_1 \frac{\1}{d} + \beta_1 \kb{\psi}{\psi} + \gamma_1 \sum_{i=1}^d \kb{i}{i} \bk{i}{\psi} \bk{\psi}{i} \right) \ket{\psi}  \; \mathrm{d} \psi \nonumber \\
&=& \frac{\alpha_1}{d} + \beta_1 + \gamma_1 \sum_{i=1}^d \int_{\norm{\psi}=1} \bk{\psi}{i}\bk{i}{\psi} \bk{i}{\psi}\bk{\psi}{i}\; \mathrm{d}\psi. \nonumber
\eea
The integral can be rewritten to give
\bea
&& \int_{\norm{\psi}=1} \bra{\psi \otimes \psi} \left( \kb{i}{i} \otimes \kb{i}{i} \right) \ket{\psi\otimes \psi} \; \mathrm{d}\psi \nonumber \\
&=& \int_{U(d)} \bra{00} \left( U \otimes U \right) \left( \kb{i}{i} \otimes \kb{i}{i} \right) \left( U \otimes U \right)^\ast \ket{00} \; \mathrm{d}U \nonumber \\
&=& \bra{00} \frac{\1+\F}{d(d+1)}\ket{00} \nonumber \\
&=& \frac{2}{d(d+1)}, \nonumber
\eea 
where $\F$ is the flip operator defined as $\F \ket{ij}= \ket{ji}$ and $dU$ denotes the normalized Haar measure on the unitary group $U(d)$ acting on $\C^d$. Together with $\gamma_1 = 1- \alpha_1 - \beta_1$, this gives an average fidelity 
\begin{equation}
\bar{f} = \frac{2}{d+1} - \alpha_1 \frac{d-1}{d(d+1)} + \beta_1 \frac{d-1}{d+1}. \nonumber
\end{equation}
Using Eqs.~(\ref{eq:paramred1},\ref{eq:paramred2}) we can express this in terms of the state $\varrho$.  We can again w.l.o.g. assume that $\tr{\varrho f_2}=0$ and $\tr{\1_2\varrho}=1-\tr{\varrho f_1}$ from the proof of Prop.~\ref{prop:cone}. Therefore, we obtain
\begin{equation}
\bar{f} = \frac{1}{d+1}\left( 1 + d\bra{e_1} \varrho \ket{e_1} \right). 
\label{eq:Avfminvarrho}
\end{equation}
We would like to maximize Eq.~(\ref{eq:Avfminvarrho}) for each value of the worst-case total variational distance, which by Lemma~\ref{lem:TV} and Eq.~(\ref{eq:paramred3}) is 
\be \delta_{TV}=1-\langle e_2|\varrho| e_2\rangle-\tr{\varrho f_1}.\nonumber \ee
Similarly to the worst-case fidelity, we can again consider the map $\varrho\mapsto\varrho+\epsilon|e_2\rangle\langle e_2|-\epsilon f_1$, $\epsilon\geq 0$, under which $\delta_{TV}$ is constant and $\bar{f}$ non-decreasing, such that $\tr{\varrho f_1}=0$ can be assumed. That is, $z=0$ is sufficient for the optimal tradeoff.

The remaining optimization problem can be solved by realizing that $(\bar{f}(d+1)-1)/d = \bra{e_1} \varrho \ket{e_1}$ and using the solution to the quadratic problem stated and solved in the worst-case fidelity tradeoff. This yields the solution stated in  Eq.~(\ref{eq:optavfidtv}).
\end{proof}

%%%%%%%%%%%%%%%%%%%%%%
\subsection{Trace norm}
The analogue of the total variational distance for density operators is  (up to a factor of $2$) the trace norm distance. The corresponding distance between a channel $T_1$ and the identity map is then given by half of the $1$-to-$1$-norm distance
\be \Delta_{TV}(T_1):=\frac12\sup_{\rho}||T_1(\rho)-\rho||_1,\label{eq:Delta11def}\ee
where the supremum is taken over all density operators. $\Delta_{TV}$ quantifies how well $T_1$ can be distinguished from $\id$ in a statistical experiment, if no ancillary system is allowed. For the two-parameter family of channels in Eq.~(\ref{eq:commutant}) $ \Delta_{TV}$ turns out to be a function of the worst-case fidelity $f$, which was defined in Eq.~(\ref{eq:worstfT1}). This is in contrast to the case of general channels, which merely satisfy the Fuchs-van de Graaf inequalities \be\label{eq:FuchsGraaf}1-f\leq \Delta_{TV}\leq\sqrt{1-f}.\ee
\begin{lemma} For every channel of the form in Eq.~(\ref{eq:commutant}), we have $\Delta_{TV}=1-f$.\label{lem:TV1f}
\end{lemma}
\begin{proof}
Due to convexity of the norm we can restrict the supremum in Eq.~(\ref{eq:Delta11def}) to pure state density operators. The resulting operator $T_1\big(|\psi\rangle\langle\psi |\big)-|\psi\rangle\langle\psi |$  then has a single negative eigenvalue and vanishing trace. Hence, the trace-norm is twice the operator norm and we can write
\bea
\Delta_{TV}(T_1) &=&  \max_{||\psi||=||\phi||=1} \langle\phi| \big[|\psi\rangle\langle\psi |-T_1\big(|\psi\rangle\langle\psi |\big)\big]|\phi\rangle\label{eq:tnfid1}\\
&=&  \max_{||\psi||=||\phi||=1}\left[(1-\beta_1)|\langle\psi|\phi\rangle|^2-\frac{\alpha_1}{d}-\gamma_1\sum_{i=1}^d|\langle\phi|i\rangle|^2|\langle\psi|i\rangle|^2\right]\nonumber\\
&=&  \max_{||\psi||=||\phi||=1} \langle\psi\otimes\phi|R|\psi\otimes\phi\rangle-\frac{\alpha_1}{d},\nonumber \\
&&  R:=(1-\beta_1)\mathbbm{F}-\gamma_1\sum_{i=1}^d|ii\rangle\langle ii|. \nonumber
\eea
Our aim is to prove that the maximum in Eq.~(\ref{eq:tnfid1}) is attained for $\psi=\phi$ since then the Lemma follows from the definition of the worst-case fidelity $f$. In order to achieve this, we exploit the symmetry properties of $R$, which is block-diagonal w.r.t. the decomposition of $\C^d\otimes\C^d$ into symmetric and anti-symmetric subspace. Moreover, if we denote by $P_+:=(\1+\mathbbm{F})/2$ the projector onto the symmetric subspace, then $R\leq P_+ R P_+$. Defining $\cS$ as the set of separable density operators and utilizing its convexity, we obtain
\bea \max_{||\psi||=||\phi||=1}\langle\psi\otimes\phi|R|\psi\otimes\phi\rangle & =  &\nonumber\max_{\rho\in\cS}\tr{R\rho} \leq \max_{\rho\in\cS}\tr{RP_+\rho P_+}\\ &=&\max_{\rho\in P_+ \cS P_+}\tr{R\rho}\nonumber\\
&=&\max_{||\psi||=1}\langle\psi\otimes\psi|R|\psi\otimes\psi\rangle,\nonumber
\eea
where the last step follows from the fact that the extreme points of the convex set $P_+\cS P_+$ are pure, symmetric product states.
\end{proof}
Due to Prop.~\ref{prop:sym} we can now plug the previous Lemma into Thm.~\ref{thm:TVfidelity} and obtain:
\begin{corollary}[Total variation - trace norm tradeoff]\label{cor:TVTV}
Consider a non-degenerate von Neumann measurement, given by an orthonormal basis in $\C^d$, and an instrument with $d$ corresponding outcomes. Then the worst-case total variational distance $\delta_{TV}$ and its trace-norm analogue $\Delta_{TV}$ satisfy
\be \delta_{TV}\geq \left\{\begin{array}{ll}\frac{1}{d}\left|\sqrt{(1-\Delta_{TV})(d-1)}-\sqrt{\Delta_{TV}} \right|^2&\ \text{if }\ \Delta_{TV}\leq1-\frac{1}{d},\\ 0 &\ \text{if }\ \Delta_{TV}\geq 1-\frac{1}{d}.
\end{array}\right.\label{eq:opttracetv}\ee
The inequality is tight and equality is attainable within the one-parameter family of instruments in Eq.~(\ref{eq:optinst}) with $z=0$.
\end{corollary}

%%%%%%%%%%%%%%%%%%%%%%%%%%%%%%%%%%%%%%%%
\subsection{Diamond norm}\label{sec:gvN}
We treat the diamond norm separately, not only because it might be the operationally most relevant measure, but also because the corresponding tradeoff result will be proven in a more general setting: we will allow the target measurement to be a von Neumann measurement that may be degenerate. We will see that degeneracy, even if it varies among the measurement outcomes, does not affect the optimal tradeoff curve if the diamond norm is considered. For general distance measures $\Delta$ that satisfy Assumption~\ref{assum:1} we do not expect this result to be true since, loosely speaking, they typically behave less benign w.r.t. extending the system than the diamond norm. Hence, assigning different dimensions to different measurement outcomes may, in general, affect the optimal information-disturbance relation. Before we prove that this is not the case for the tradeoff between the diamond norm and its classical counterpart, the total variational distance, let us recall its definition and basic properties.

For a hermiticity-preserving map $\Phi:\cM_d\rightarrow\cM_{d'}$ we define 
\be ||\Phi||_\diamond:=\sup_{\rho}||(\Phi\otimes\id_d)(\rho)||_1,\ee
where the supremum is taken over all density operators in $\cM_{d^2}$, which by convexity may be assumed to be pure. For a
 quantum channel $T_1:\cM_d\rightarrow\cM_d$ we then define 
\be \Delta_\diamond(T_1):=||T_1-\id_d||_\diamond.\ee
 $\Delta_\diamond(T_1)$ quantifies how well $T_1$ can be distinguished from the identity channel $\id$ in a statistical experiment, when arbitrary preparations, measurements and ancillary systems are allowed. There are two crucial properties of the diamond norm that we will exploit: 1) Monotonicity:  for any quantum channel $\Psi$, neither $||\Psi\circ\Phi||_\diamond$  nor $||\Phi\circ\Psi||_\diamond$ can be larger than $||\Phi||_\diamond$. 2) Tensor stability: in particular, $||\Phi\otimes\id||_\diamond=||\Phi||_\diamond$.
 
 \begin{lemma}[Dimension-independence of optimal tradeoff curve]\label{lem:dimind}
 Consider a von Neumann measurement with $m$ outcomes,  corresponding to $m$ mutually orthogonal, non-zero projections of possibly different dimensions, as target. Then the optimal $\Delta_\diamond - \delta_{TV}$-tradeoff depends only on $m$ and is independent of the dimensions of the projections.
 \end{lemma}
\begin{proof}
Let $(d_1,\ldots,d_m)\in\mathbbm{N}^m$ be the dimensions of the projections  (i.e., the dimensions of their ranges) and assume w.l.o.g. that $d_m$ is the largest among them. We will consider three changes of those dimensions, namely
\be (1,\ldots, 1)\rightarrow(d_m\ldots,d_m)\rightarrow (d_1,\ldots, d_m)\rightarrow (1,\ldots, 1),\label{eq:dimchanges} 
\ee 
and show that in each of those three steps the accessible region in the $\Delta_{\diamond}-\delta_{TV}$-plane can only grow or stay the same. Since Eq.~(\ref{eq:dimchanges}) describes a full circle, this means that the region, indeed, stays the same, which proves the claim of the Lemma.

For the starting point in Eq.~(\ref{eq:dimchanges}) we consider an arbitrary instrument $\big(I_i:\cM_m\rightarrow\cM_m\big)_{i=1}^m$ that is supposed to approximate a von Neumann measurement given by $(|i\rangle\langle i|)_{i=1}^m$. From here, we construct an instrument that approximates $(|i\rangle\langle i|\otimes\1_{d_m})_{i=1}^m$ simply by taking $ I_i\otimes{\id_{d_m}}=: \tilde{I}_i $. Then $\Delta_\diamond\big(\sum_i\tilde{I}_i\big)=\Delta_\diamond\big(\sum_i I_i\big)$ holds due to the tensor stability of the diamond norm and 
\bea && \sup_\rho \sum_{i=1}^m \Big|\tr{\rho\big(\tilde{I}^*_i(\1)- |i\rangle\langle i|\otimes\1_{d_m}\big)}\Big| \nonumber \\
&=& \sup_\rho \sum_{i=1}^m \Big|\tr{\rho\left( \big(I^*_i(\1)- |i\rangle\langle i|\big)\otimes\1_{d_m}\right)}\Big|  \nonumber \\ 
&=& \sup_\rho\sum_{i=1}^m \Big|\tr{\rho\big(I^*_i(\1)- |i\rangle\langle i|\big)}\Big| \nonumber
\eea shows that the value of $\delta_{TV}$ is preserved, as well.

Second and third step in Eq.~(\ref{eq:dimchanges}) can be treated at once by realizing that in both cases the dimensions are pointwise non-increasing. So let us consider this scenario in general. Denote the projections corresponding to two von Neumann measurements by $Q_i\in\cM_D$ and $\tilde{Q}_i\in\cM_{\tilde{D}}$ and assume that $\tr{Q_i}=:d_i\geq \tilde{d}_i:=\tr{\tilde{Q}_i}$. Let $I_i:\cM_D\rightarrow\cM_D$ be the elements of an instrument that approximates the measurement in the larger space. In order to construct an instrument in the smaller space that is at least as good w.r.t. $\Delta_\diamond$ and $\delta$, we introduce two isometries $V$ and $W$ as
\bea V:\C^{\tilde{D}}\rightarrow\C^D &\text{ s.t. }& V^*Q_i V=\tilde{Q}_i \nonumber \\
W:\C^D\rightarrow\C^k\otimes\C^{\tilde{D}} &\text{ s.t. }& \forall i\in\{1,\ldots,\tilde{D}\}:\ WV|i\rangle=|1\rangle\otimes |i\rangle, \nonumber
\eea
where $\{|i\rangle\}_{i}$ is an orthonormal basis in $\C^{\tilde{D}}$ and $k\in\mathbbm{N}$ is sufficiently large so that $W$ can be an isometry. The sought instrument in the smaller space can then be defined as
\be \tilde{I}_i(\rho):={\rm tr}_{\C^k} \left[W I_i\big(V\rho V^*\big)W^*\right], \nonumber
\ee
where ${\rm tr}_{\C^k}$ means the partial trace w.r.t. the first tensor factor.
For the value of $\Delta_\diamond$ we obtain
\bea \Big|\!\Big| \id-\sum_i  \tilde{I}_i \Big|\!\Big|_\diamond &=& \Big|\!\Big|{\rm tr}_{\C^k}\big[WV\cdot V^* W^*\big]-{\rm tr}_{\C^k}\Big[W\Big(\sum_i I_i\big(V\cdot V^*\big)\Big)W^*\Big] \Big|\!\Big|_\diamond\nonumber\\
&\leq &\Big|\!\Big| V\cdot V^* -\sum_i I_i\big(V\cdot V^*\big)\Big|\!\Big|_\diamond\ \leq \ \Big|\!\Big|\id-\sum_i I_i\Big|\!\Big|_\diamond\nonumber ,
\eea
where we have used the monotonicity property of the diamond norm twice.
Finally, using that $\tilde{I}_i^*(\1)=V^* I_i^*(\1)V$ we can show that also  $\delta_{TV}$ is non-increasing when moving to the smaller space since 
\bea
\sup_{\rho} \sum_i\left|\tr{\rho\big(\tilde{I}_i^*(\1)-\tilde{Q}_i\big)}\right| &=& \sup_{\rho} \sum_i\left|\tr{V \rho V^*\big(I_i^*(\1)-Q_i\big)}\right| \nonumber\\
&\leq& \sup_{\rho} \sum_i\left|\tr{ \rho \big(I_i^*(\1)-Q_i\big)}\right| ,\nonumber
\eea
where the supremum in the first (second) line is taken over all density operators in the smaller (larger) space.
\end{proof} 
 
 \begin{theorem}[Total variation - diamond norm tradeoff]\label{thm:TVdiamond}
If an instrument is considered approximating a (possibly degenerate) von Neumann measurement with $m$ outcomes, then the worst-case total variational distance $\delta_{TV}$ and the diamond norm distance $\Delta_\diamond$ satisfy
\be \delta_{TV}\geq \left\{\begin{array}{ll}\frac{1}{2m}\left(\sqrt{(2-\Delta_\diamond)(m-1)}-\sqrt{\Delta_\diamond} \right)^2&\ \text{if }\ \Delta_\diamond\leq 2-\frac{2}{m},\\ 0 &\ \text{if }\ \Delta_\diamond > 2-\frac{2}{m}.
\end{array}\right.\label{eq:optdiamondtv}\ee
The inequality is tight in the sense that for every choice of the von Neumann measurement there is an instrument achieving equality.
\end{theorem}
Note: if the von Neumann measurement is non-degenerate, then equality is again attainable within the one-parameter family of instruments in Eq.~(\ref{eq:optinst}) with $z=0$. In the degenerate case, equality is attainable by such instruments when suitably embedded, as it is done in the proof of Lemma~\ref{lem:dimind}.
\begin{figure}[ht]
	\centering
		\includegraphics[clip, trim=4cm 9cm 4cm 9cm, width=0.900\textwidth]{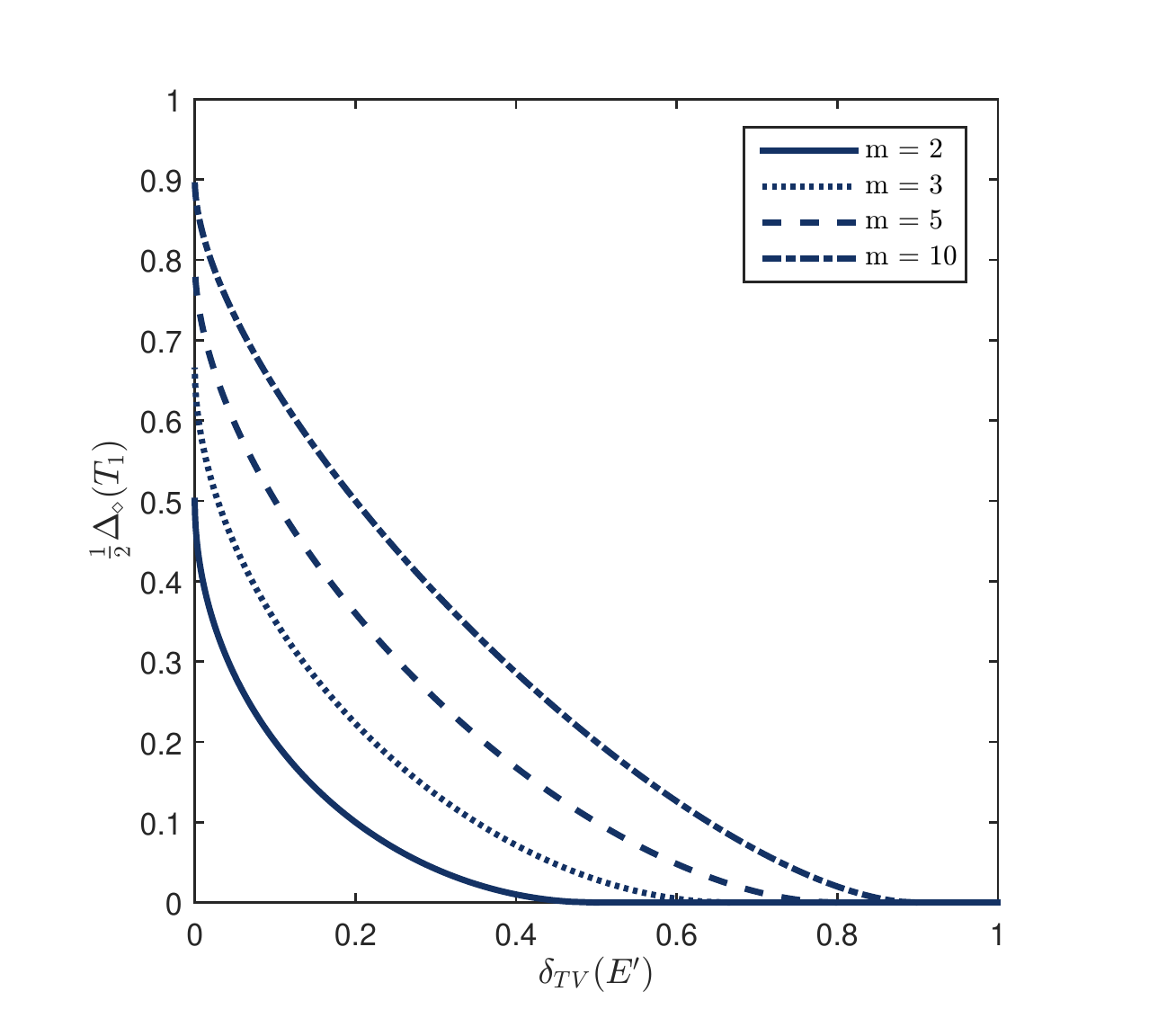}
			\caption{The optimal total variation - diamond norm tradeoff for different numbers of measurement outcome.}
	\label{fig:Plot_TV_Diamond_Tradeoff}
\end{figure}
\begin{proof}
Due to Lemma~\ref{lem:dimind} we can assume that the von Neumann measurement is non-degenerate and acts on a $d=m$  dimensional Hilbert space. We will prove that the accessible region stays the same when replacing $\Delta_\diamond$ with $2\Delta_{TV}$ so that the theorem follows from Cor.~\ref{cor:TVTV}.

 Since $\Delta_\diamond \geq 2\Delta_{TV}$ it suffices to show that this holds with equality for instruments that achieve the optimal $\Delta_{TV} - \delta_{TV}$ curve. Due to  Eq.~(\ref{eq:alpha1z}) and Cor.~\ref{cor:TVTV} we can restrict ourselves to symmetric channels $T_1$ of the form in Eq.~(\ref{eq:commutant}) with $\alpha_1=0$. With $\cC(\cdot):=\sum_{i=1}^d |i\rangle\langle i|\langle i|\cdot|i\rangle$ and using that $(1-\beta_1)=\gamma_1$ we have 
\bea \Delta_\diamond(T_1) &=& \sup_{||\psi||=1} \big|\!\big| \big(T_1\otimes\id_d-\id_{d^2}\big)\big(|\psi\rangle\langle\psi|\big)\big|\!\big|_1\nonumber\\
&=& \sup_{||\psi||=1} \gamma_1 \big|\!\big| |\psi\rangle\langle\psi| -\big(\cC\otimes\id_d\big)\big(|\psi\rangle\langle\psi|\big)\big|\!\big|_1\nonumber\\
&=& 2 \gamma_1 \sup_{||\psi||=||\phi||=1} |\langle\psi|\phi\rangle|^2-\langle\phi|\big(\cC\otimes\id_d\big)\big(|\psi\rangle\langle\psi|\big)|\phi\rangle\nonumber\\
&=& 2\gamma_1 \sup_{||\psi||=1} 1-\langle\psi|\big(\cC\otimes\id_d\big)\big(|\psi\rangle\langle\psi|\big)|\psi\rangle,\nonumber
\eea
where the last two steps follow exactly the argumentation below Eq.~(\ref{eq:tnfid1}). For the remaining optimization problem we write $|\psi\rangle=(\1_d\otimes X)\sum_{i=1}^d |ii\rangle$ where $X\in\cM_d$ is s.t. $\sum_{i=1}^d \langle i|X^*X|i\rangle=||\psi||^2=1$. Then 
$$ \langle\psi|\big(\cC\otimes\id_d\big)\big(|\psi\rangle\langle\psi|\big)|\psi\rangle = \sum_{i=1}^d \big|\langle i|X^*X|i\rangle\big|^2 \geq \frac1d\Big(\sum_{i=1}^d \langle i|X^*X|i\rangle\Big)^2=\frac1d ,$$
where the inequality is an application of Cauchy-Schwarz. Consequently, 
\be\label{eq:cbfinal} \Delta_\diamond (T_1)\leq 2\gamma_1\Big(1-\frac1d\Big)=2\Delta_{TV}(T_1),\ee
where the last inequality uses that $\Delta_{TV}=1-f$ by Lemma~\ref{lem:TV1f} and $f=1-\gamma(1-1/d)$ by Eq.~(\ref{eq:fgamma}). As $\Delta_\diamond$ is also lower bounded by $2 \Delta_{TV}$, equality has to hold in Eq.~(\ref{eq:cbfinal}), which completes the proof. \end{proof}
Note that equality in Eq.~(\ref{eq:cbfinal}) means that entanglement assistance does not increase the distinguishability of the identity channel $\id$ and the channel $T_1$.

%%%%%%%%%%%%%%%%%%%%%%%%%%%%%%%%%%%%%%%%
\section{SDPs for general POVMs}\label{sec:SDP}
%%%%%%%%%%%%%%%%%%%%%%%%%%%%%%%%%%%%%%%%
In this section, we consider the most general case, when the target measurement $E$ is given by an arbitrary POVM. It is then still possible to characterize the achievable region in the $\Delta-\delta$-plane as the set of solutions to some SDP if $\Delta$ and $\delta$ are convex semialgebraic. To this end, let us start with the definition of semialgebraicity.  

A semialgebraic set is a set $S\subseteq\R^n$ defined by a finite sequence of polynomial equations and inequalities or any finite union of such sets. We mainly follow \cite{Bochnak_RealAlgGeo, Karow_2003}.
\begin{definition}[Semialgebraic set {\cite[Definition 3.1.1]{Karow_2003}}]
A semialgebraic subset of $\R^n$ is an element of the Boolean algebra of subsets of $\R^n$ which is generated by the sets 
\begin{equation}
\left\{ \left( x_1, \ldots, x_n \right) \in \R^n \middle\vert p\left(x_1, \ldots, x_n \right) > 0 \right\}, \ \ p\in \R[X_1, \ldots, X_n],
\label{eq:semiset}
\end{equation}
where $\R[X_1, \ldots, X_n]$ denotes the ring of real polynomials in the variables $X_1$, $\ldots$, $X_n$.
\end{definition}
From this definition, it is immediately clear that sets of the form \[\left\{ \left( x_1, \ldots, x_n \right) \in \R^n \middle\vert p\left(x_1, \ldots, x_n \right) \bullet 0 \right\},\] where $\bullet \in \{ <,>, \leq, \geq, =, \neq \}$, $p\in \R[X_1, \ldots, X_n]$, are semialgebraic and that the family of semialgebraic sets is closed under taking complements, finite unions and finite intersections. Moreover, by the Tarski-Seidenberg principle quantification over reals preserves the semialgebraic property \cite[Appendix 1]{Marshall_2008}:
\begin{theorem}[Tarski-Seidenberg, quantifier elimination {\cite[Thm. 1]{Wolf_2011}}]
\label{thm:TS2} 
Given a finite set $\{ p_i (x,z)\}_{i=1}^k$ of polynomial equalities and inequalities with variables $(x,z)\in \R^n \times \R^m$ and coefficients in $\Q$. Let $\phi(x,z)$ be a Boolean combination of the $p_i$'s (using $\vee$, $\wedge$ and $\neg$) and 
\begin{equation}
\Psi(z) := \big( Q_1 x_1 \ldots Q_n x_n : \phi(x,z) \big), \ \ Q_j \in \left\{ \exists, \forall \right\}.
\label{eq:TS1}
\end{equation}
Then there exists a formula $\psi(z)$ which is (i) a quantifier-free Boolean combination of finitely many polynomial (in-)equalities with rational coefficients, and (ii) equivalent in the sense
\begin{equation}
\forall z: \ \ \big( \psi(z) \Leftrightarrow \Psi(z) \big).
\label{eq:TS2}
\end{equation}
Moreover, there exists an effective algorithm which constructs the quantifier-free equivalent $\psi$ of any such formula $\Psi$.
\end{theorem}

\begin{definition}[Semialgebraic function]
Let $S_k \subseteq \R^{n_k}$ be non-empty semialgebraic sets, $k=1,2$. A function $f:S_1 \to S_2$ is said to be semialgebraic if its graph 
\begin{equation}
\left\{(x,z) \in S_1\times S_2 \middle\vert z=f(x)  \right\}
\label{eq:graph}
\end{equation} 
is a semialgebraic subset of $\R^{n_1+n_2}$.
\end{definition}

Using the Tarski-Seidenberg principle, Thm.~\ref{thm:TS2}, it is also possible to prove that the following functions, that are likely to appear in optimization problems, are semialgebraic \cite[Sec. 3.1]{Karow_2003}:
\begin{itemize}
	\item Real polynomial functions are semialgebraic.
	\item Compositions of semialgebraic functions are semialgebraic. Let $S_k \subseteq \R^{n_k}$, $k=1,2,3$, be semialgebraic sets and let $f:S_1 \to S_2$ and $g:S_2 \to S_3$ be semialgebaric functions. Then their composition $g \circ f:S_1 \to S_3$ is semialgebraic.
	%\item Scalar products are semialgebraic.
	%\item Indicator functions of semialgebraic set are semialgebraic.
  \item Let $f: S_1 \to S_2$ be a semialgebraic function, and let $A \subseteq S_1$  (resp. $B \subseteq S_2$) be a semialgebraic set. Then $f(A)$ (resp. $f^{-1}(B)$) is semialgebraic. 
	\item Finite sums and products of semialgebraic functions are semialgebraic. Let $f_1, f_2: S_1 \to \R$ be semialgebraic functions. Then $f_1+f_2, f_1f_2:S_1 \to \R$ are semialgebraic.
	\item Let $f_1, f_2: S_1 \to \R$ be semialgebraic functions. If $f_2^{-1}(\{0\}) \neq S_1$, then $f_1/f_2:S_1\backslash f_2^{-1}(\{0 \}) \to \R$ is semialgebraic.
	\item Let $\cM^{\text{Herm}}_n$ denote the set of all Hermitian $n\times n$-matrices, and for $H \in \cM^{\text{Herm}}_n$ let $\lambda_k(H)$, $k \in \{1,\ldots, n\}$, denote the eigenvalues of $H$ in decreasing order. The functions $\lambda_k(\cdot):\cM^{\text{Herm}}_n \to \R$ are semialgebraic.
	\item The singular value functions $\sigma_k: \C^{m\times n} \to [0, \infty)$, $1\leq k \leq \min \{m,n\}$ are semialgebraic.
\end{itemize}

For the last point, we identify a subset of $\C^{n}$ with a subset of $\R^{2n}$ by separating the real and imaginary parts. Therefore, the notion of a semialgebraic subset of $\C^{m\times n}$  is well defined. 

Furthermore, one can show the following regarding the supremum or infimum of a function:
\begin{lemma}[{\cite[Cor. 3.1.15]{Karow_2003}}]
\label{lem:supinf}
Let $S_k \subseteq \R^{n_k}$ be non-empty semialgebraic sets, $k=1,2$, and $f:S_1 \times S_2 \to \R$ a semialgebraic function. Then $\hat{f}, \check{f}:S_1 \to \R \cup \{-\infty, \infty\}$,
\begin{eqnarray}
\hat{f}(x) &:=& \sup_{y\in S_2} f(x,y) \ \ \text{ and} \\
\check{f}(x) &:=& \inf_{y \in S_2} f(x,y) 
\end{eqnarray}
are both semialgebraic.
\end{lemma}

Using the fact that singular value functions are semialgebraic, it is immediately possible to show the following corollary:
\begin{corollary}[{\cite[Cor. 3.1.24]{Karow_2003}}]
\label{cor:Schatten}
The Schatten $p$-norms $\norm{\cdot}_{p}:\C^{n\times m} \to [0,\infty)$ are semialgebraic for all $p \in [1,\infty) \cap \Q$ and $p=\infty$.
\end{corollary}
\begin{proof}
Please see \cite[Cor. 3.1.23 and 3.1.19]{Karow_2003} for a full proof. The main idea is to establish that the function $x \mapsto x^{p/q}$, with $x > 0$ and $p,q$ positive integers, is semialgebraic. Its graph is 
\begin{eqnarray*}
&& \left\{ \left(x,z \right) \in \R^2_+ \middle\vert z = x^{\frac{p}{q}} \right\} \\
&=&  \left\{ \left(x,z \right) \in \R^2 \middle\vert z^q - x^{p} =0 \right\}  \cap \R^2_+,
\end{eqnarray*}
which is semialgebraic.
\end{proof}

\begin{corollary}
\label{cor:Schattenpq}
The Schatten $p$-to-$q$ norm-distances of a quantum channel $\Phi \in \cT_d$ to the identity channel
$$\Phi\ \mapsto ||\Phi-\id||_{p\rightarrow q,n}:=\sup_{\rho\in\cS_{dn}}\frac{||(\Phi-\id)\otimes\id_n(\rho)||_q}{||\rho||_p},\quad  n\in\mathbbm{N},$$ 
are semialgebraic for all  $p,q \in [1,\infty) \cap \Q$ and $p,q=\infty$.

The worst-case fidelity distance of a quantum channel $\Phi \in \cT_d$ to the identity channel
$$\Phi\ \mapsto \inf_{\rho\in\cS_{d}} F\left( \Phi(\rho), \rho\right)^2$$
is semialgebraic.

The worst-case $l_p$-distances of a POVM $E' \in \cE_{d,m}$ to the target POVM $E \in \cE_{d,m}$
$$E'\ \mapsto \sup_{\rho\in\cS_{d}}||\left( \Tr[\rho E_i] - \Tr[\rho E_i']\right)_{i=1}^m ||_{p},$$
are semialgebraic for all  $p \in [1,\infty) \cap \Q$ and $p=\infty$.
\end{corollary}
\begin{proof}
Given that the set of all quantum states is semialgebraic \cite[Lemma 1]{Wolf_2011}, Cor.~\ref{cor:Schatten} together with Lemma~\ref{lem:supinf} immediately yields the statements.
\end{proof}
In particular, the special case of the \emph{diamond norm} $||\cdot||_\diamond:=||\cdot||_{1\rightarrow1,d}$, which we discuss in more  detail below,  and its dual, the \emph{cb-norm} (with $p=q=\infty, n=d$) are semialgebraic.

\begin{theorem}[Helton-Nie conjecture in dimension two {\cite[Thm.~6.8.]{Scheiderer_2012}}]
\label{thm:HeltonNie}
Every convex semialgebraic subset $S$ of $\R^2$ is the feasible set of a SDP. That is, it can be written as
\begin{equation}
S = \left\{ \xi \in \R^2\middle\vert \exists \eta \in \R^m : A + \sum_{i=1}^2 \xi_i B_i + \sum_{j=1}^m \eta_j C_j \geq 0 \right\},
\label{eq:SDPFeasible}
\end{equation}
where $m \geq 0$ and $A$, $B_i$ as well as $C_j$ are real symmetric matrices of the same size. 
\end{theorem}
The proof of the Helton-Nie conjecture in dimension two can be found in \cite{Scheiderer_2012}.\footnote{The conjecture for larger dimensions was shown to be false in general in \cite{Scheiderer_2017}.} The main observation of this section is a consequence of the previous theorem and the following simple Lemma:

\begin{lemma}
\label{lem:SemiPlane}
If $\Delta$ and $\delta$ are both semialgebraic, then the accessible region in the $\Delta-\delta$-plane is a semialgebraic set.
\end{lemma}
\begin{proof}
Let us denote the accessible region in the $\Delta-\delta$-plane by $S$, i.e.,
\begin{equation*}
S = \left\{ x\in \R^2 \middle\vert \exists I = \{ I_i \}_{i=1}^m : x_1 = \Delta\left(\sum_{i=1}^{m} I_i \right) \wedge x_2 = \delta \left( \left( I_i^\ast(\1)\right)_{i=1}^m\right) \right\}.
\end{equation*}
First note that the set of instruments is semialgebraic. 
The maps $I\mapsto \sum_{i=1}^m I_i$ as well as $I \mapsto (I_i^\ast(\1))_{i=1}^m$ are algebraic and therefore semialgebraic \cite{Bochnak_RealAlgGeo}.  Given that the composition of two semialgebraic maps is semialgebraic \cite[Prop. 2.2.6 (i)]{Bochnak_RealAlgGeo} and that the image of a semialgebraic set under a semialgebraic map is semialgebraic \cite[Prop. 2.2.7.]{Bochnak_RealAlgGeo}, $\Delta\left(\sum_{i=1}^m I_i\right)$ as well as $\delta\left((I_i^\ast(\1))_{i=1}^m\right)$ are semialgebraic. Using the Tarski-Seidenberg principle, Thm.~\ref{thm:TS2}, we arrive at the claim.
\end{proof}

\begin{theorem}[SDP solution for arbitrary target measurements] 
\label{thm:SDPalg}
If $\Delta$ and $\delta$ are both convex and semialgebraic, then the accessible region in the $\Delta-\delta$-plane is the feasible set of a SDP.
\end{theorem}
\begin{proof}
If $\Delta$ and $\delta$ are convex and semialgebraic, then the whole region in the $\Delta-\delta$-plane that is accessible by quantum instruments is a convex semialgebraic subset of $\R^2$ by Lemma~\ref{lem:SemiPlane}. By Thm.~\ref{thm:HeltonNie}, it must thus be the feasible set of a SDP.
\end{proof}

In particular, if we consider a Schatten $p$-to-$q$-norm distance, with $p$ and $q$ rational, to describe the disturbance caused to the quantum system and a worst-case $l_p$-norm distances, with rational $p$, to quantify the measurement error, the accessible region in the $\Delta-\delta$-plane is the feasible set of a SDP. 

Unfortunately, we do not know how to make the results of \cite{Scheiderer_2012} constructive. That is while Thm.~\ref{thm:SDPalg} proves the existence of a SDP, we do not have a way of making the SDP explicit. 
 
\vspace*{5pt}
\paragraph{\bf SDP for the diamond norm tradeoff} 
We now explicitly state the SDP yielding the optimal tradeoff curve in the case of a general POVM for the worst-case $l_\infty$-distance and the diamond norm. 
This particular example does not rely on the general result of Thm~\ref{thm:SDPalg}, since the $l_\infty$-norm as well as the diamond norm are already well-suited to SDP formulation.
 Please note that on the measurement error side, we use the worst-case $l_\infty$-norm to quantify the distance between the two probability distributions,
\begin{equation}
\delta_{l_\infty} := \sup_{\rho}\max_i \Big|\tr{E_i'\rho}- \tr{E_i\rho}\Big|. 
\end{equation}
In this setting the optimization problem, quantifying the information-disturbance tradeoff, is given as: \\ 
Compute for a given target POVM $E = \left\{ E_i \right\}_{i=1}^m$ and $\lambda \in \left[ 0, 1 \right]$
\bea
\label{eq:OptProb}
	\nu (E,\lambda)  := &  & \min_{\left\{ I_i \right\}_{i=1}^m} \norm{\sum^m_{i=1} I_i - \id}_\diamond  \\
	& \text{such that}  & \norm{I^\ast_i(\1)-E_i}_\infty \leq \lambda \ \ \forall i, \nonumber \\
	& & I_i \text{ is completely positive} \ \  \forall i \text{ and }  \nonumber \\
	& & \sum^m_{i=1} I^\ast_i(\1) = \1. \nonumber
\eea
In the following, let us the define the Choi matrix for any linear map $T: \cM_d \to \cM_{d'}$ as 
\begin{equation}
J(T):= \left( T \otimes \id_d \right) \left( \sum_{i,j=1}^d \kb{ii}{jj} \right).
\label{eq:J}
\end{equation}
\begin{theorem}
\label{thm:SDP}
For a given target POVM $E = \left\{ E_i \in \cM_d \right\}_{i=1}^m$ and $\lambda \in \left[0, 1\right]$, the optimization problem $\nu (E, \lambda)$ given in Eq.~(\ref{eq:OptProb}), can be formulated as a SDP $(\phi, C, D)$, where $\phi:\cM_{\hat{d}} \to \cM_{\check{d}}$ is a hermiticity preserving map, $C=C^\ast \in \cM_{\hat{d}}$ and $D=D^\ast \in \cM_{\check{d}}$, with dimensions $\hat{d} = (m+4)d^2+2(m+2)d$ and $\check{d} = 2+(m+2)d^2$. The primal and the dual SDP problem are given as follows: \\
\begin{equation*}
  \begin{split}
                  \text{\emph{Primal SDP problem}} & \\ & \\
    \text{maximize } \  \ & \tr{CX} \\
    \text{subject to } \ \ &
    \begin{aligned}[t]
&\phi(X) = D \\
&X \geq 0
    \end{aligned}
  \end{split}
  \qquad \qquad 
  \begin{split}
                  \text{\emph{Dual SDP problem}} & \\ & \\
    \text{minimize } \ \ & \tr{DY} \\
    \text{subject to } \ \ &
    \begin{aligned}[t]
&\phi^\ast(Y) \geq C \\
&Y = Y^\dagger
    \end{aligned}
  \end{split}
\end{equation*}
where the hermiticity preserving map $\phi:\cM_{\hat{d}} \to \cM_{\check{d}}$ is 
\bea
\phi(X) & = & \tr{w_0} \oplus \tr{w_1} \oplus \left(A+Z_0-\1 \otimes w_0\right) \oplus \left(B+Z_1-\1 \otimes w_1\right) \oplus \nonumber \\
&& \bigoplus_{i=1}^m \left( M+M^\ast + \1 \otimes \left(F_i-\widetilde{F}_i \right) +G_i+\1 \otimes \left(H-\widetilde{H}\right)\right),
\eea
with 
\bea
X & := & \begin{pmatrix}
	A & M \\
	M^\ast & B
\end{pmatrix} 
\oplus w_0 \oplus w_1 \oplus Z_0 \oplus Z_1 \oplus  \nonumber \\
& & \bigoplus_{i=1}^m F_i \oplus  \bigoplus_{i=1}^m \widetilde{F}_i \oplus \bigoplus_{i=1}^m G_i \oplus H \oplus \widetilde{H}.
\eea
The adjoint of the map $\phi$ is 
\bea
\phi^\ast(Y) &:=&  \begin{pmatrix}
	Y_0 & \sum_{i=1}^m J(I_i) \\
	\sum_{i=1}^m J(I_i) & Y_1
\end{pmatrix}  \oplus  \left( \lambda_0 \1 - \Tr_{1} \left[Y_0\right] \right)  \oplus  \left( \lambda_1 \1 - \Tr_{1}\left[ Y_1\right] \right) \oplus \nonumber \\
& &   Y_0  \oplus Y_1  \oplus \bigoplus_{i=1}^m \Tr_{1} \left[J(I_i)\right]  \oplus \bigoplus_{i=1}^m -\Tr_{1} \left[J(I_i)\right] \oplus \bigoplus_{i=1}^m J(I_i) \oplus \nonumber\\
& & \sum^m_{i=1} \Tr_{1}\left[J(I_i)\right]  \oplus  -\sum^m_{i=1} \Tr_{1}\left[J(I_i)\right],
\eea
with
\begin{equation}
Y :=  \lambda_0  \oplus  \lambda_1  \oplus  Y_0  \oplus  Y_1  \oplus \bigoplus_{i=1}^m   J(I_i). 
\end{equation}
Furthermore,
\begin{equation}
D:=  \frac{1}{2}  \oplus  \frac{1}{2}  \oplus  0  \oplus   0   \oplus \bigoplus_{i=1}^m 0
\end{equation}
and
\bea
C & := & \begin{pmatrix}
	0 & J(\id) \\ 
	J(\id) & 0
\end{pmatrix}  \oplus  0  \oplus  0 \oplus  0  \oplus  0  \oplus \bigoplus_{i=1}^m \left(-\lambda \1 +E_i^T \right)  \oplus \nonumber \\ 
& &  \bigoplus_{i=1}^m \left(-\lambda \1 -E_i^T \right)  \oplus \bigoplus_{i=1}^m 0  \oplus  \1  \oplus  -\1. 
\eea
\end{theorem}
\begin{proof}
The diamond norm can be expressed as a SDP itself \cite{Watrous_2009, Watrous_2012}, 
\begin{align*}
\norm{\id - \sum_{i=1}^m I_i}_\diamond = & && \min_{Y_0, Y_1 \in \cM_d \otimes \cM_d} \frac{1}{2} \left[ \norm{\Tr_{1}\left[ Y_0 \right]}_\infty + \norm{\Tr_1\left[ Y_1 \right]}_\infty  \right] \\
& \text{ such that } &&
\begin{pmatrix}
	Y_0 & J\left( \id - \sum_{i=1}^m I_i \right) \\
	J\left( \id - \sum^m_{i=1} I_i \right) & Y_1 
\end{pmatrix} 
\geq 0 \ \ \text{ and} \nonumber \\
&&& Y_0,Y_1 \geq 0, \nonumber
\end{align*}
where $\Tr_1$ denotes the partial trace over the first system.
Using Watrous SDP for the diamond norm in the form of \cite[p. 11]{Watrous_2012} gives 
\begin{align*}
\nu \left( E, \lambda \right) = & \text{ minimize } && \frac{1}{2}\left[ \lambda_0 +\lambda_1 \right]  \\
& \text{ such that } &&
\begin{pmatrix}
	Y_0 & \sum_{i=1}^m J(I_i) \\
	\sum_{i=1}^m J(I_i) & Y_1
\end{pmatrix} 
\geq
\begin{pmatrix}
	0 & J(\id) \\
	J(\id) & 0
\end{pmatrix} \nonumber \\
&&& \lambda_0 \1 - \Tr_{1} \left[Y_0 \right] \geq 0 \nonumber \\
&&& \lambda_1 \1 - \Tr_{1} \left[Y_1 \right] \geq 0 \nonumber \\
&&& Y_0, Y_1 \geq 0 \nonumber \\
&&& \Tr_{1}\left[J(I_i)\right] \geq -\lambda \1 + E_i^T \ \ \forall i \nonumber \\
&&& -\Tr_{1}\left[J(I_i)\right] \geq -\lambda \1 - E_i^T \ \ \forall i \nonumber \\
&&& J(I_i) \geq 0 \ \ \forall i \nonumber \\
&&& \sum_{i=1}^m \Tr_{1}\left[J(I_i)\right] \geq \1 \nonumber \\
&&& -\sum_{i=1}^m \Tr_{1}\left[J(I_i)\right] \geq -\1.  \nonumber
\end{align*}
We would like to write this as a SDP in the form
\begin{align*}
\text{minimize } \ \ & \tr{DY} \\
    \text{subject to } \ \ 
&\phi^\ast(Y) \geq C, \\
&Y = Y^\dagger.
\end{align*}
 Collecting all variables that we optimize over yields $Y \in \C \oplus \C \oplus \cM_{d^2} \oplus \cM_{d^2} \oplus \bigoplus_{i=1}^m \cM_{d^2}$ as
\begin{equation*}
Y :=  \lambda_0  \oplus  \lambda_1  \oplus  Y_0  \oplus  Y_1  \oplus \bigoplus_{i=1}^m   J(I_i). 
\end{equation*}
Furthermore, we set $D \in \C \oplus \C \oplus \cM_{d^2} \oplus \cM_{d^2} \oplus \bigoplus_{i=1}^m \cM_{d^2}$ as
\begin{equation*}
D:=  \frac{1}{2}  \oplus  \frac{1}{2}  \oplus  0_{d^2}  \oplus   0_{d^2}   \oplus \bigoplus_{i=1}^m 0_{d^2}.
\end{equation*}
Similarly, set $\phi^\ast(Y) \in \cM_{2d^2} \oplus \cM_d \oplus \cM_d \oplus \cM_{d^2}  \oplus \cM_{d^2} \oplus \bigoplus_{i=1}^m \cM_d \oplus \bigoplus_{i=1}^m \cM_d \oplus \bigoplus_{i=1}^m \cM_{d^2} \oplus \bigoplus_{i=1}^m \cM_d  \oplus \bigoplus_{i=1}^m \cM_d$ to be
\bea
\phi^\ast(Y) &:=&  \begin{pmatrix}
	Y_0 & \sum_{i=1}^m J(I_i) \\
	\sum_{i=1}^m J(I_i) & Y_1
\end{pmatrix}  \oplus  \left( \lambda_0 \1_d - \Tr_{1} \left[Y_0\right] \right)  \oplus  \left( \lambda_1 \1_d - \Tr_{1}\left[ Y_1\right] \right) \oplus \nonumber \\
& &   Y_0  \oplus Y_1  \oplus \bigoplus_{i=1}^m \Tr_{1} \left[J(I_i)\right]  \oplus \bigoplus_{i=1}^m -\Tr_{1} \left[J(I_i)\right] \oplus \bigoplus_{i=1}^m J(I_i) \oplus \nonumber\\
& & \sum^m_{i=1} \Tr_{1}\left[J(I_i)\right]  \oplus  -\sum^m_{i=1} \Tr_{1}\left[J(I_i)\right],\nonumber
\eea
and we define $C \in \cM_{2d^2} \oplus \cM_d \oplus \cM_d \oplus \cM_{d^2}  \oplus \cM_{d^2} \oplus \bigoplus_{i=1}^m \cM_d \oplus \bigoplus_{i=1}^m \cM_d \oplus \bigoplus_{i=1}^m \cM_{d^2} \oplus \bigoplus_{i=1}^m \cM_d  \oplus \bigoplus_{i=1}^m \cM_d$ as
\bea
C & := & \begin{pmatrix}
	0 & J(\id) \\ 
	J(\id) & 0
\end{pmatrix}  \oplus  0_d  \oplus  0_d \oplus  0_{d^2}  \oplus  0_{d^2}  \oplus \bigoplus_{i=1}^m \left(-\lambda \1 +E_i^T \right)  \oplus \nonumber \\ 
& &  \bigoplus_{i=1}^m \left(-\lambda \1 -E_i^T \right)  \oplus \bigoplus_{i=1}^m 0_{d^2}  \oplus  \1_d  \oplus  -\1_d. \nonumber
\eea
Therefore, the optimization problem $\nu(E,\lambda)$ is a SDP indeed. In order to state the dual SDP problem, define $X \in \cM_{2d^2} \oplus \cM_d \oplus \cM_d \oplus \cM_{d^2}  \oplus \cM_{d^2} \oplus \bigoplus_{i=1}^m \cM_d \oplus \bigoplus_{i=1}^m \cM_d \oplus \bigoplus_{i=1}^m \cM_{d^2} \oplus \bigoplus_{i=1}^m \cM_d  \oplus \bigoplus_{i=1}^m \cM_d$ to be
\bea
X & := & \begin{pmatrix}
	A & M \\
	M^\ast & B
\end{pmatrix} 
\oplus w_0 \oplus w_1 \oplus Z_0 \oplus Z_1 \oplus  \nonumber \\
& & \bigoplus_{i=1}^m F_i \oplus  \bigoplus_{i=1}^m \widetilde{F}_i \oplus \bigoplus_{i=1}^m G_i \oplus H \oplus \widetilde{H}.\nonumber
\eea
Using the fact that $\tr{\phi^\ast(Y)X} = \tr{Y\phi(X)}$ lets us construct $\phi$ such that $\phi(X) \in \C \oplus \C \oplus \cM_{d^2} \oplus \cM_{d^2} \oplus \bigoplus_{i=1}^m \cM_{d^2}$ is
\bea
\phi(X) & = & \tr{w_0} \oplus \tr{w_1} \oplus \left(A+Z_0-\1 \otimes w_0\right) \oplus \left(B+Z_1-\1 \otimes w_1\right) \oplus \nonumber \\
&& \bigoplus_{i=1}^m \left( M+M^\ast + \1 \otimes \left(F_i-\widetilde{F}_i \right) +G_i+\1 \otimes \left(H-\widetilde{H}\right)\right).\nonumber
\eea
\end{proof}
\begin{proposition}
For the above SDP $\left(\phi, C, D\right)$ the Slater-type strong duality holds, such that
\begin{equation}
\sup_X \tr{CX} = \inf_Y \tr{DY}.
\end{equation}
\end{proposition}
\begin{proof}
There is an interior point $X > 0$ that fulfills $\phi(X) = D$ and a $Y=Y^\ast$ such that $\phi^\ast(Y) \geq C$. By Slater's theorem strong duality holds for the SDP $\left(\phi, C, D\right)$.
\end{proof}

Using Thm.~\ref{thm:SDP} it is therefore possible to explicitly state the SDP that yields the information-disturbance tradeoff curve for any general POVM in the case where the measurement-error is quantified by the worst-case $l_\infty$-distance and the disturbance is quantified by the diamond norm. 
\begin{figure}[ht]
	\centering
		\includegraphics[clip, trim=4cm 9cm 4cm 9cm, width=0.900\textwidth]{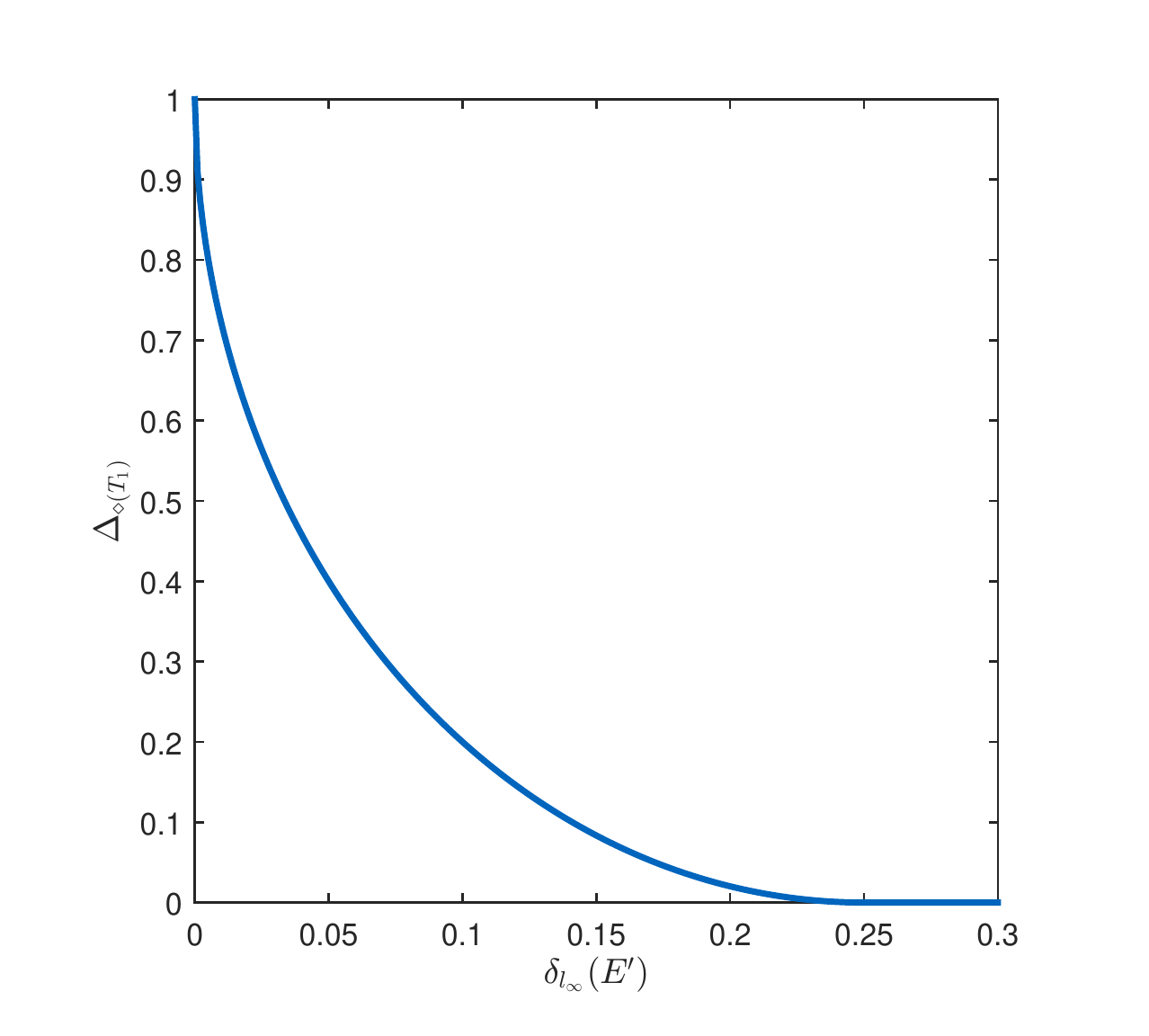}
	\caption{The information-disturbance tradeoff for a qubit SIC POVM target measurement.}
	\label{fig:SDP_dim2_4_SIC_POVM}
\end{figure}
\begin{figure}[ht]
	\centering
		\includegraphics[clip, trim=4cm 9cm 4cm 9cm, width=0.900\textwidth]{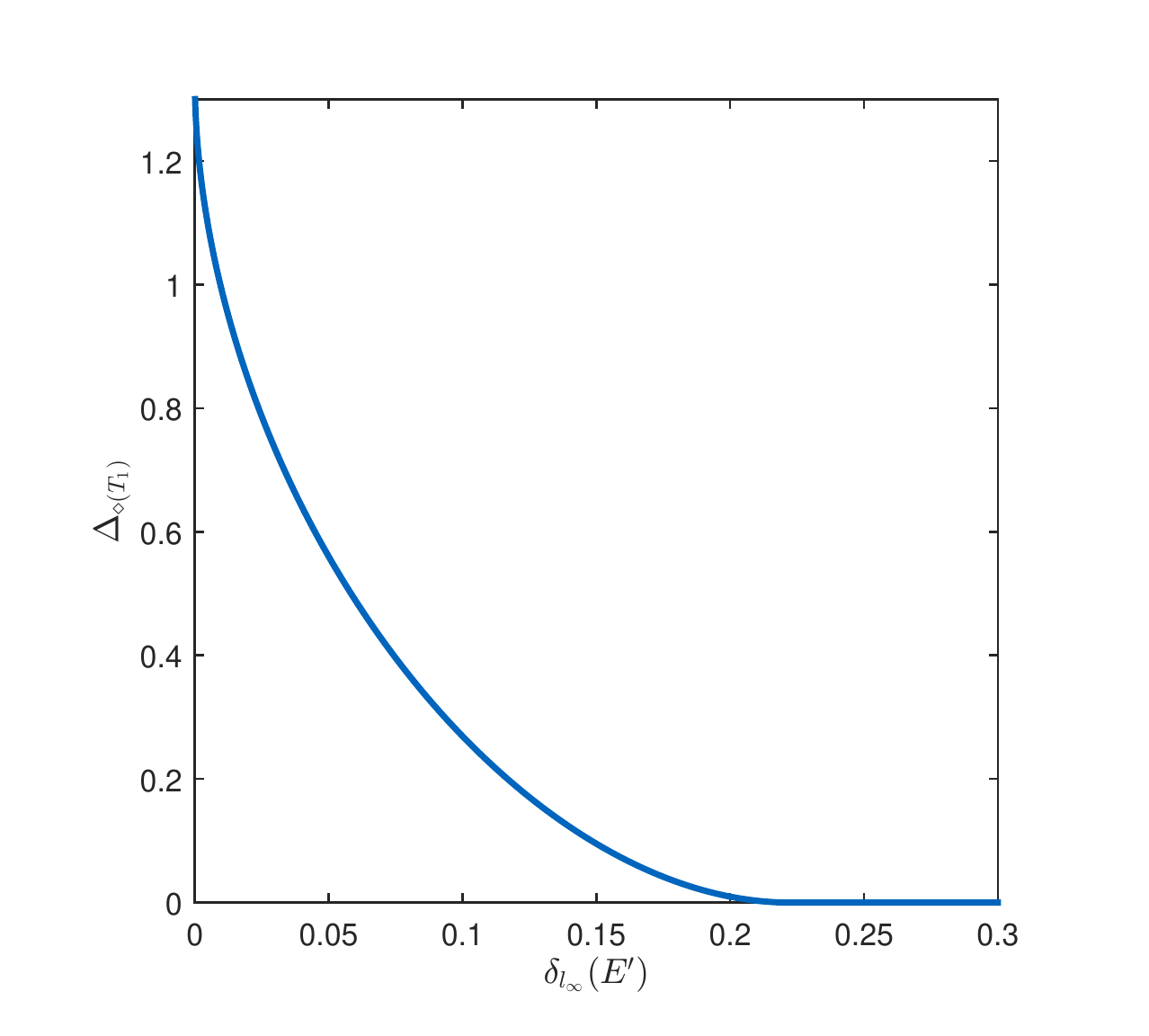}
	\caption{The information-disturbance tradeoff for a qutrit SIC POVM target measurement.}
	\label{fig:SDP_dim3_9_SIC_POVM}
\end{figure}
\vspace*{5pt}
\paragraph{\bf SIC POVM} 
As it is a prominent application in various fields in quantum information theory, this section analyzes the example of a symmetric, informationally complete (SIC) POVM as target measurement. A SIC POVM is defined by a set of $d^2$ subnormalized rank-$1$ projectors $\left\{P_i/d\right\}_{i=1}^{d^2}$, which have equal pairwise Hilbert-Schmidt inner products, $\tr{P_iP_j}/d^2=1/d^2(d+1)$ for $i \neq j$. Figure~\ref{fig:SDP_dim2_4_SIC_POVM} and ~\ref{fig:SDP_dim3_9_SIC_POVM} show the information-disturbance tradeoff for a qubit SIC POVM and qutrit SIC POVM as target measurement respectively. In two dimensions, we considered the following SIC POVM represented by the four Bloch vectors $(0,0,1)$, $(2\sqrt{2}/3, 0,-1/3)$, $(-\sqrt{2}/3, \sqrt{2/3},-1/3)$ and $(-\sqrt{2}/3, -\sqrt{2/3},-1/3)$. In dimension three, the nine explicit (unnormalized) vectors of the SIC POVM under consideration are $(0,1,-1)$, $(0,1,-\eta)$, $(0,1,-\eta^2)$, $(-1,0,1)$, $(-\eta,0,1)$, $(-\eta^2,0,1)$, $(1,-1,0)$, $(1,-\eta,0)$ and $(1, -\eta^2,0)$ with $\eta = \exp{2\pi i/3}$. To solve the SDP stated in Thm.~\ref{thm:SDP} for this particular example, we used cvx, a package for specifying and solving convex programs \cite{cvx, Grant_Boyd_2008} in {MATLAB} \cite{Matlab}. 

The solution of the SDP is compared to an instrument similar to the one found in Thm.~\ref{thm:universality} consisting of an inherit POVM $E' = tE+(1-t) \1/d$, $t\in [0,1]$, together with the L\"uders channel. The symmetry of the SIC POVM most likely leads to this agreement. However, further investigation would be necessary to get a better understanding of this observation.

\section*{Acknowledgment}
The authors would like to thank Teiko Heinosaari for many useful comments. 
AKHs work is supported by the Elite Network of Bavaria through the PhD program of excellence \textit{Exploring Quantum Matter}.
This research was supported in part by the National Science Foundation under Grant No. NSF PHY11-25915.

\newpage
\appendix

\section*{Appendix}

\paragraph{\bf Proof that average- and worst-case construction satisfy Assumption~\ref{assum:1} and Assumption~\ref{assum:2}.}

\begin{lemma}
If $\tilde{\Delta}:\cS_d\times\cS_d\rightarrow[0,\infty]$ satisfies 
\begin{enumerate}[(i)]
	\item $\tilde{\Delta}(\rho,\rho)=0$,
	\item convexity in its first argument and
	\item unitary invariance,
\end{enumerate}
then the  worst-case as well as the average-case construction
\begin{eqnarray*}
\Delta_\infty(\Phi)&:=& \sup_{\rho\in S} \tilde{\Delta}\left(\Phi(\rho),\rho\right) \ \ \text{ and}\\
\Delta_\mu(\Phi)&:=& \int_{\cS_d} \tilde{\Delta}\left(\Phi(\rho),\rho\right)\; \mathrm{d}\mu(\rho), 
\end{eqnarray*}
with $\mu$ a unitarily invariant measure on $\cS_d$ and $S\subseteq\cS_d$ a unitarily closed subset,
satisfy Assumption~\ref{assum:1}.
\end{lemma}
\begin{proof}
Let $\tilde{\Delta}:\cS_d\times\cS_d\rightarrow[0,\infty]$ be such that it
\begin{enumerate}[(i)]
	\item satisfies $\tilde{\Delta}(\rho,\rho)=0$,
	\item is convex in its first argument, i.e., for any quantum state $\sigma, \sigma', \rho \in \cS_d$
	\be
	\tilde{\Delta}\left( \lambda \sigma + (1-\lambda) \sigma' , \rho \right) \leq \lambda \tilde{\Delta} \left(\sigma, \rho \right) + (1-\lambda) \tilde{\Delta} \left( \sigma', \rho \right) \ \ \forall \lambda \in \left[0,1\right], \nonumber
	\ee
	\item and is unitarily invariant, i.e., for any quantum state $\sigma,  \rho \in \cS_d$
	\be
	\tilde{\Delta}\left( U^\ast \sigma U, U^\ast \rho U \right) =  \tilde{\Delta}\left( \sigma, \rho \right) \ \ \forall \text{ unitaries } U \in \cM_d. \nonumber
	\ee
\end{enumerate}
Then its worst case $\Delta_\infty$ satisfies
\begin{enumerate}[(a)]
	\item $\Delta_\infty(\id) = 0$, since
	\begin{equation}
	\Delta_\infty(\id) = \sup_{\rho\in S} \tilde{\Delta}\left(\id(\rho),\rho\right) = \sup_{\rho\in S} \tilde{\Delta}\left(\rho,\rho\right) = 0,\nonumber
	\end{equation}
	\item is convex, i.e., for every quantum channel $\Phi, \Phi' \in \cT_d$
	\begin{equation}
	\Delta_\infty \left( \lambda \Phi + (1-\lambda) \Phi' \right) \leq \lambda \Delta_\infty \left( \Phi \right) + (1-\lambda) \Delta_\infty \left( \Phi' \right) \ \ \forall \lambda \in \left[0,1\right],\nonumber
	\end{equation}
	because
	\begin{eqnarray*}
	\Delta_\infty \left( \lambda \Phi + (1-\lambda) \Phi'\right) &=& \sup_{\rho\in S} \tilde{\Delta}\left(\lambda \Phi(\rho) + (1-\lambda) \Phi'(\rho),\rho\right) \\
	&\leq& \lambda \sup_{\rho\in S} \tilde{\Delta}\left( \Phi(\rho), \rho \right) + (1-\lambda) \sup_{\rho\in S} \left(\Phi'(\rho),\rho\right) \\
	&=& \lambda \Delta_\infty \left( \Phi \right) + (1-\lambda) \Delta_\infty \left( \Phi' \right),
	\end{eqnarray*}
	\item and is basis-independent, i.e., for every unitary $U \in \cM_d$ and every channel $\Phi \in \cT_d$, we have that
	\begin{equation}
	\Delta_\infty \left(U \Phi\left(U^\ast \cdot U \right) U^\ast \right) = \Delta_\infty(\Phi),\nonumber
	\end{equation}
	since
	\begin{eqnarray*}
	\Delta_\infty \left(U \Phi\left(U^\ast \cdot U \right) U^\ast \right) &=& \sup_{\rho\in S} \tilde{\Delta}\left(U \Phi\left(U^\ast \rho U \right) U^\ast ,\rho\right) \\
	&=& \sup_{\rho\in S} \tilde{\Delta}\left( \Phi\left(U^\ast \rho U \right)  ,U^\ast\rho U \right) \\
	&=& \sup_{\rho\in S} \tilde{\Delta}\left( \Phi\left( \rho  \right)  ,\rho  \right) \\
	&=& \Delta_\infty(\Phi).
	\end{eqnarray*}
\end{enumerate}
The average case $\Delta_\mu$ satisfies 
\begin{enumerate}[(a)]
	\item $\Delta_\mu(\id) = 0$, since
	\begin{equation}
	\Delta_\mu(\id) = \int_{\cS_d} \tilde{\Delta}\left(\id(\rho),\rho\right)\; \mathrm{d}\mu(\rho) = \int_{\cS_d}  \tilde{\Delta}\left(\rho,\rho\right)\; \mathrm{d}\mu(\rho) = 0,\nonumber
	\end{equation}
	\item is convex, i.e.,  for every quantum channel $\Phi, \Phi' \in \cT_d$
	\begin{equation}
	\Delta_\mu \left( \lambda \Phi + (1-\lambda) \Phi' \right) \leq \lambda \Delta_\mu \left( \Phi \right) + (1-\lambda) \Delta_\mu \left( \Phi' \right) \ \ \forall \lambda \in \left[0,1\right],\nonumber
	\end{equation}
	because
	\begin{eqnarray*}
	\Delta_\mu \left( \lambda \Phi + (1-\lambda) \Phi'\right) &=& \int_{\cS_d}  \tilde{\Delta}\left(\lambda \Phi(\rho) + (1-\lambda) \Phi'(\rho),\rho\right)\; \mathrm{d}\mu(\rho) \\
	&\leq& \lambda \int_{\cS_d} \tilde{\Delta}\left( \Phi(\rho), \rho \right) \; \mathrm{d}\mu(\rho) + (1-\lambda) \int_{\cS_d} \tilde{\Delta} \left(\Phi'(\rho),\rho\right) \; \mathrm{d}\mu(\rho)\\
	&=& \lambda \Delta_\mu \left( \Phi \right) + (1-\lambda) \Delta_\mu \left( \Phi' \right),
	\end{eqnarray*}
	\item and is basis-independent, i.e., for every unitary $U \in \cM_d$ and every channel $\Phi \in \cT_d$, we have that
	\begin{equation}
	\Delta_\mu \left(U \Phi\left(U^\ast \cdot U \right) U^\ast \right) = \Delta_\mu(\Phi),\nonumber
	\end{equation}
	since
	\begin{eqnarray*}
	\Delta_\mu \left(U \Phi\left(U^\ast \cdot U \right) U^\ast \right) &=& \int_{\cS_d}  \tilde{\Delta}\left(U \Phi\left(U^\ast \rho U \right) U^\ast ,\rho\right)\; \mathrm{d}\mu(\rho) \\
	&=& \int_{\cS_d}  \tilde{\Delta}\left( \Phi\left(U^\ast \rho U \right)  ,U^\ast\rho U \right)\; \mathrm{d}\mu(\rho) \\
	&=& \int_{\cS_d}  \tilde{\Delta}\left( \Phi\left( \rho  \right)  ,\rho  \right)\; \mathrm{d}\mu(\rho) \\
	&=& \Delta_\mu(\Phi),
	\end{eqnarray*}
	where we have used the fact that $\mu$ is a unitarily invariant measure on $\cS_d$. 
\end{enumerate}
The worst-case construction as well as the average-case construction therefore satisfy Assumption~\ref{assum:1} as claimed. 
\end{proof}

\begin{lemma}
If $\tilde{\delta}:\cP_d\times\cP_d\rightarrow[0,\infty]$ on the space of probability distributions $\cP_d:=\big\{q\in\R^d|\sum_{i=1}^d q_i=1\wedge \forall i: q_i\geq 0\big\}$ applied to the target distribution $p_i:=\langle i|\rho|i\rangle$ and the actually measured distribution $p_i':=\tr{\rho E_i'}$ satisfies
\begin{enumerate}[(i)]
	\item $\tilde{\delta}(q,q)=0$,
	\item convexity in its second argument and
	\item invariance under joint permutations,
\end{enumerate}
then the worst-case as well as the average-case construction 
\bea
\delta_{\infty}(E')&:=&\sup_{\rho\in S} \tilde{\delta}(p,p'),\nonumber\\
\delta_{\mu}(E')&:=&\int_{\cS_d} \tilde{\delta}(p,p') \; \mathrm{d}\mu(\rho),\nonumber 
\eea
 both satisfy Assumption~\ref{assum:2}. 
\end{lemma}
\begin{proof}
Let $\tilde{\delta}:\cP_d\times\cP_d\rightarrow[0,\infty]$ be such that it
\begin{enumerate}[(i)]
	\item satisfies $\tilde{\delta}(q,q)=0$,
	\item is convex in its second argument, i.e., for every probability distribution $p,q,q' \in \cP_d$
	\begin{equation}
	\tilde{\delta}( p , \lambda q +  (1-\lambda) q') \leq \lambda \tilde{\delta}(p,q) +(1-\lambda) \tilde{\delta}(p,q')\ \ \forall \lambda \in [0,1],\nonumber
	\end{equation}
	\item and invariant under joint permutations,i.e., for every quantum state $\rho \in \cS_d$ and every POVM $E,E' \in \cE_d$
	\begin{equation}
	\tilde{\delta}\left( \left( \tr{\rho U_{\pi}^\ast E_{\pi(i)} U_\pi  } \right)_{i=1}^d, \left( \tr{\rho U_\pi^\ast  E_{\pi(i)}'} U_\pi \right)_{i=1}^d \right) = 	\tilde{\delta}\left( \left( \tr{\rho  E_{i}  } \right)_{i=1}^d, \left( \tr{\rho  E_{i}'  } \right)_{i=1}^d \right).\nonumber
	\end{equation}
\end{enumerate}
Then its worst case $\delta_\infty$ satisfies 
\begin{enumerate}[(a)]
	\item $\delta_\infty \left( \left( \kb{i}{i} \right)_{i=1}^d \right) = 0$, since
	\begin{equation*}
	\delta_\infty \left( \left( \kb{i}{i} \right)_{i=1}^d \right) = \sup_{\rho\in S} \tilde{\delta}\left(\left( \kb{i}{i} \right)_{i=1}^d ,\left( \kb{i}{i} \right)_{i=1}^d \right)=0,
	\end{equation*}
	\item is convex, i.e., for any POVM $Q,Q' \in \cE_d$
	\begin{equation*}
	\delta_\infty \left( \lambda Q + (1-\lambda)Q' \right) \leq 	\lambda \delta_\infty \left(  Q \right) + (1-\lambda)\delta_\infty \left( Q' \right) \ \ \forall \lambda \in [0,1],
	\end{equation*}
	because 
	\begin{eqnarray*}
		\delta_\infty \left( \lambda Q + (1-\lambda)Q' \right)  &=& \sup_{\rho\in S} \tilde{\delta}(p,\lambda q+(1-\lambda)q') \\
		&\leq& \lambda \sup_{\rho\in S}  \tilde{\delta}(p, q) + (1-\lambda) \sup_{\rho\in S}  \tilde{\delta}(p,q') \\
		&=& \lambda \delta_\infty \left(  Q \right) + (1-\lambda)\delta_\infty \left( Q' \right),
	\end{eqnarray*}
	where we have denoted the corresponding probability distribution as $q_i:= \tr{\rho Q_i}$ and $q_i':= \tr{\rho Q_i'}$. 
	\item is permutation-invariant, i.e.,	for every permutation $\pi\in S_d$ and any POVM $E \in\cE_d$
 \be 
 \delta_\infty \left(\left( U_\pi^\ast E_{\pi(i)} U_\pi \right)_{i=1}^d \right)=\delta_\infty \left( \left(E_{i}\right)_{i=1}^d \right), \nonumber
 \ee
 where $U_\pi$ is the permutation matrix that acts as $U_\pi |i\rangle=|\pi(i)\rangle$, since
 \begin{eqnarray*}
\delta_\infty \left(\left( U_\pi^\ast E_{\pi(i)} U_\pi \right)_{i=1}^d\right) &=& \sup_{\rho\in S} \tilde{\delta}\left(\left(\tr{\rho \kb{i}{i}} \right)_{i=1}^d,\left(\tr{\rho U_\pi^\ast E_{\pi(i)} U_\pi} \right)_{i=1}^d\right) \\
&=& \sup_{\rho\in S} \tilde{\delta}\left(\left(\tr{\rho \kb{i}{i}} \right)_{i=1}^d,\left(\tr{U_\pi \rho U_\pi^\ast E_{\pi(i)} } \right)_{i=1}^d\right) \\
&=& \sup_{\rho\in S} \tilde{\delta}\left(\left(\tr{U_\pi^\ast \rho U_\pi \kb{i}{i}} \right)_{i=1}^d,\left(\tr{ \rho  E_{\pi(i)} } \right)_{i=1}^d\right) \\
&=& \sup_{\rho\in S} \tilde{\delta}\left(\left(\tr{ \rho U_\pi \kb{i}{i}U_\pi^\ast} \right)_{i=1}^d,\left(\tr{ \rho  E_{\pi(i)} } \right)_{i=1}^d\right) \\
&=& \sup_{\rho\in S} \tilde{\delta}\left(\left(\tr{ \rho \kb{\pi(i)}{\pi(i)}} \right)_{i=1}^d,\left(\tr{ \rho  E_{\pi(i)} } \right)_{i=1}^d\right) \\
&=& \sup_{\rho\in S} \tilde{\delta}\left(\left(\tr{ \rho \kb{i}{i}} \right)_{i=1}^d,\left(\tr{ \rho  E_{i} } \right)_{i=1}^d\right) \\
&=& \delta_\infty\left(\left(E_{i}\right)_{i=1}^d\right),
 \end{eqnarray*}
 \item and it satisfies for every diagonal unitary $D\in\cM_d$ and any POVM $E\in\cE_d$
\be  
\delta_\infty \left( (D^\ast E_i D)_{i=1}^d  \right)=\delta_\infty \left(  (E_i)_{i=1}^d  \right), \nonumber 
\ee
because
\begin{eqnarray*}
\delta_\infty \left( (D^\ast E_i D)_{i=1}^d  \right) &=&  \sup_{\rho\in S} \tilde{\delta}\left( \left(\tr{ \rho \kb{i}{i}} \right)_{i=1}^d,\left(\tr{ \rho D^\ast E_{i} D } \right)_{i=1}^d \right) \\
&=&  \sup_{\rho\in S} \tilde{\delta}\left( \left(\tr{ \rho \kb{i}{i}} \right)_{i=1}^d,\left(\tr{ D \rho D^\ast E_{i}  } \right)_{i=1}^d \right) \\
&=&  \sup_{\rho\in S} \tilde{\delta}\left( \left(\tr{ D \rho D^\ast \kb{i}{i}} \right)_{i=1}^d,\left(\tr{  \rho  E_{i}  } \right)_{i=1}^d \right) \\
&=&  \sup_{\rho\in S} \tilde{\delta}\left( \left(\tr{  \rho D^\ast \kb{i}{i}D} \right)_{i=1}^d,\left(\tr{  \rho  E_{i}  } \right)_{i=1}^d \right) \\
&=&  \sup_{\rho\in S} \tilde{\delta}\left( \left(\tr{  \rho \kb{i}{i}} \right)_{i=1}^d,\left(\tr{  \rho  E_{i}  } \right)_{i=1}^d \right) \\
&=& \delta_\infty \left(  (E_i)_{i=1}^d  \right).
\end{eqnarray*}
\end{enumerate}
Similarly, its average case $\delta_\mu$ satisfies 
\begin{enumerate}[(a)]
	\item $\delta_\mu \left( \left( \kb{i}{i} \right)_{i=1}^d \right) = 0$, since
	\begin{equation*}
	\delta_\mu \left( \left( \kb{i}{i} \right)_{i=1}^d \right) = \int_{\cS_d}  \tilde{\delta}\left(\left( \kb{i}{i} \right)_{i=1}^d ,\left( \kb{i}{i} \right)_{i=1}^d \right)\; \mathrm{d}\mu(\rho)=0,
	\end{equation*}
	\item is convex, i.e., for any POVM $Q,Q' \in \cE_d$
	\begin{equation*}
	\delta_\mu \left( \lambda Q + (1-\lambda)Q' \right) \leq 	\lambda \delta_\mu \left(  Q \right) + (1-\lambda)\delta_\mu \left( Q' \right) \ \ \forall \lambda \in [0,1],
	\end{equation*}
	because 
	\begin{eqnarray*}
		\delta_\mu \left( \lambda Q + (1-\lambda)Q' \right)  &=& \int_{\cS_d}  \tilde{\delta}(p,\lambda q+(1-\lambda)q')\; \mathrm{d}\mu(\rho)  \\
		&\leq& \lambda \int_{\cS_d}  \tilde{\delta}(p, q) \; d\mu(\rho) + (1-\lambda)\int_{\cS_d}   \tilde{\delta}(p,q')\; \mathrm{d}\mu(\rho) \\
		&=& \lambda \delta_\mu \left(  Q \right) + (1-\lambda)\delta_\mu \left( Q' \right),
	\end{eqnarray*}
	where we have denoted the corresponding probability distribution as $q_i:= \tr{\rho Q_i}$ and $q_i':= \tr{\rho Q_i'}$.
	\item is permutation-invariant, i.e.	for every permutation $\pi\in S_d$ and any $E\in\cE_d$
 \be 
 \delta_\mu \left(\left( U_\pi^\ast E_{\pi(i)} U_\pi \right)_{i=1}^d \right)=\delta_\mu \left( \left(E_{i}\right)_{i=1}^d \right) \nonumber
 \ee
 where $U_\pi$ is the permutation matrix that acts as $U_\pi |i\rangle=|\pi(i)\rangle$, since
 \begin{eqnarray*}
\delta_\mu \left(\left( U_\pi^\ast E_{\pi(i)} U_\pi \right)_{i=1}^d\right) &=& \int_{\cS_d} \tilde{\delta}\left(\left(\tr{\rho \kb{i}{i}} \right)_{i=1}^d,\left(\tr{\rho U_\pi^\ast E_{\pi(i)} U_\pi} \right)_{i=1}^d\right) \; \mathrm{d}\mu(\rho)  \\
&=& \int_{\cS_d}  \tilde{\delta}\left(\left(\tr{\rho \kb{i}{i}} \right)_{i=1}^d,\left(\tr{U_\pi \rho U_\pi^\ast E_{\pi(i)} } \right)_{i=1}^d\right) \; \mathrm{d}\mu(\rho)\\
&=& \int_{\cS_d}   \tilde{\delta}\left(\left(\tr{U_\pi^\ast \rho U_\pi \kb{i}{i}} \right)_{i=1}^d,\left(\tr{ \rho  E_{\pi(i)} } \right)_{i=1}^d\right)\; \mathrm{d}\mu(\rho) \\
&=& \int_{\cS_d}  \tilde{\delta}\left(\left(\tr{ \rho U_\pi \kb{i}{i}U_\pi^\ast} \right)_{i=1}^d,\left(\tr{ \rho  E_{\pi(i)} } \right)_{i=1}^d\right) \; \mathrm{d}\mu(\rho) \\
&=& \int_{\cS_d}   \tilde{\delta}\left(\left(\tr{ \rho \kb{\pi(i)}{\pi(i)}} \right)_{i=1}^d,\left(\tr{ \rho  E_{\pi(i)} } \right)_{i=1}^d\right)\; \mathrm{d}\mu(\rho) \\
&=& \int_{\cS_d} \tilde{\delta}\left(\left(\tr{ \rho \kb{i}{i}} \right)_{i=1}^d,\left(\tr{ \rho  E_{i} } \right)_{i=1}^d\right)  \; \mathrm{d}\mu(\rho) \\
&=& \delta_\mu \left(\left(E_{i}\right)_{i=1}^d\right),
 \end{eqnarray*}
 \item and it satisfies for every diagonal unitary $D\in\cM_d$ and any $E\in\cE_d$
\be  
\delta_\mu \left( (D^\ast E_i D)_{i=1}^d  \right)=\delta_\mu \left(  (E_i)_{i=1}^d  \right), \nonumber 
\ee
because
\begin{eqnarray*}
\delta_\mu \left( (D^\ast E_i D)_{i=1}^d  \right) &=&  \int_{\cS_d}  \tilde{\delta}\left( \left(\tr{ \rho \kb{i}{i}} \right)_{i=1}^d,\left(\tr{ \rho D^\ast E_{i} D } \right)_{i=1}^d \right)\; \mathrm{d}\mu(\rho)  \\
&=&  \int_{\cS_d}   \tilde{\delta}\left( \left(\tr{ \rho \kb{i}{i}} \right)_{i=1}^d,\left(\tr{ D \rho D^\ast E_{i}  } \right)_{i=1}^d \right) \; \mathrm{d}\mu(\rho)\\
&=&  \int_{\cS_d}  \tilde{\delta}\left( \left(\tr{ D \rho D^\ast \kb{i}{i}} \right)_{i=1}^d,\left(\tr{  \rho  E_{i}  } \right)_{i=1}^d \right) \; \mathrm{d}\mu(\rho) \\
&=&  \int_{\cS_d}   \tilde{\delta}\left( \left(\tr{  \rho D^\ast \kb{i}{i}D} \right)_{i=1}^d,\left(\tr{  \rho  E_{i}  } \right)_{i=1}^d \right) \; \mathrm{d}\mu(\rho)\\
&=&  \int_{\cS_d}   \tilde{\delta}\left( \left(\tr{  \rho \kb{i}{i}} \right)_{i=1}^d,\left(\tr{  \rho  E_{i}  } \right)_{i=1}^d \right) \; \mathrm{d}\mu(\rho)\\
&=& \delta_\mu \left(  (E_i)_{i=1}^d  \right).
\end{eqnarray*}
\end{enumerate}
The worst-case as well as the average-case construction therefore satisfy Assumption~\ref{assum:2}.
\end{proof}

\paragraph{\bf Proof of Corollary~\ref{cor:reduc}.}

\begin{proof}
The eigenvalues of $J_{T_i}$, $i=1,2$, can be obtained from the expectation values of the mutually orthogonal projectors, i.e.,
\be 
x_1=\frac{\tr{(P_x\otimes\1) J_T}}{\tr{P_x}}\quad\text{and}\quad x_2=\frac{\tr{(\1\otimes P_x) J_T}}{\tr{P_x}},\quad x\in\{a,b,c\}. \nonumber
\ee
Since we know that $a,b,c$ are related to $\alpha,\beta,\gamma$ via $\alpha=d^2 a,\ \beta=b-c,\ \gamma=d(c-a)$, we get
\begin{eqnarray}
\alpha_1 &=& d^2 \frac{\tr{(P_a\otimes\1) J_T}}{\tr{P_a}} \nonumber \\
&=& d^2 \frac{\tr{\left( \1_{d^3} - \sum_{i=1}^d |ii\rangle\langle ii| \otimes \1_d \right) J_T}}{d^2-d}. \nonumber
\end{eqnarray}
Similarly, 
\begin{eqnarray}
\beta_1 &=& \frac{\tr{(P_b\otimes\1) J_T}}{\tr{P_b}} - \frac{\tr{(P_c\otimes\1) J_T}}{\tr{P_c}} \nonumber \\
&=& \frac{\tr{\left(\frac1d\sum_{i,j=1}^d |ii\rangle\langle jj|\otimes\1_d \right) J_T}}{1} \nonumber \\
&& - \frac{\tr{\left(\sum_{i=1}^d  |ii\rangle\langle ii| \otimes \1_d - \frac1d\sum_{i,j=1}^d |ii\rangle\langle jj| \otimes \1_d \right) J_T}}{d-1}, \nonumber
\end{eqnarray}
and
\begin{eqnarray}
a_2 &=& \frac{\tr{(\1\otimes P_a) J_T}}{\tr{P_a}} \nonumber \\
&=& \frac{\tr{\left(\1_{d^3}-\1_d \otimes \sum_{i=1}^d |ii\rangle\langle ii|\right) J_T}}{d^2-d}. \nonumber
\end{eqnarray}
Using the diagrammatic notation introduced earlier, i.e.,
\begin{align*}
\1_{d^3} &=:\ \Ga & \1_d\otimes\sum_{i=1}^d |ii\rangle\langle ii| &=:\ \Gb\\
\sum_{i,j=1}^d |ii\rangle\langle jj|\otimes\1_d &=:\ \Gc & \sum_{i=1}^d |ii\rangle\langle ii|\otimes\1_d &=:\ \Gd
\end{align*}
together with the isomorphic representation from Lemma~\ref{lem:iso}, the claim follows immediately. 
\end{proof}

\bibliography{Information_Disturbance_Tradeoff_Literature}{}
\bibliographystyle{ieeetr}

\end{document}